\documentclass{article}
\usepackage[pdftex]{graphicx,xcolor}
\usepackage[fleqn]{amsmath}
\usepackage{newtxtext}
\usepackage[varg]{newtxmath}

\usepackage{color}
\newcommand{\red}[1]{{#1}}

\title{Reconfiguring $k$-Path Vertex Covers}

\author{%
 Duc A. Hoang$^1$ and
 Akira Suzuki$^2$ and
 Tsuyoshi Yagita$^3$
}

\date{
$^1$ Graduate School of Informatics, Kyoto University, Japan\\ \texttt{hoang.duc.8r@kyoto-u.ac.jp}\\
$^2$ Graduate School of Information Sciences, Tohoku University, Japan\\ \texttt{akira@tohoku.ac.jp}\\
$^3$ Graduate School of Computer Science and Systems Engineering, Kyushu Institute of Technology, Japan\\ \texttt{yagita.tsuyoshi307@mail.kyutech.jp}
}

\usepackage{amsthm}

\newtheorem{theorem}{Theorem}

\newtheorem{lemma}[theorem]{Lemma}

\usepackage{tikz}
\usetikzlibrary{calc}

\usepackage[linesnumbered,ruled,noend]{algorithm2e}

\begin{document}
\maketitle
\begin{abstract}
A vertex subset $I$ of a graph $G$ is called a $k$-path vertex cover if every path on $k$ vertices in $G$ contains at least one vertex from $I$.
The \textsc{$k$-Path Vertex Cover Reconfiguration ($k$-PVCR)} problem asks if one can transform one $k$-path vertex cover into another via a sequence of $k$-path vertex covers where each intermediate member is obtained from its predecessor by applying a given reconfiguration rule exactly once.
We investigate the computational complexity of \textsc{$k$-PVCR} from the viewpoint of graph classes under the well-known reconfiguration rules: $\mathsf{TS}$, $\mathsf{TJ}$, and $\mathsf{TAR}$.
The problem for $k=2$, known as the \textsc{Vertex Cover Reconfiguration (VCR)} problem, has been well-studied in the literature.
We show that certain known hardness results for \textsc{VCR} on different graph classes 
can be extended for \textsc{$k$-PVCR}.
In particular, we prove a complexity dichotomy for \textsc{$k$-PVCR} on general graphs: on those whose maximum degree is three (and even planar), the problem is $\mathtt{PSPACE}$-complete, while on those whose maximum degree is two (i.e., paths and cycles), the problem can be solved in polynomial time.
Additionally, we also design polynomial-time algorithms for \textsc{$k$-PVCR} on trees under each of $\mathsf{TJ}$ and $\mathsf{TAR}$.
Moreover, on paths, cycles, and trees, we describe how one can construct a reconfiguration sequence between two given $k$-path vertex covers in a yes-instance.
In particular, on paths, our constructed reconfiguration sequence is shortest.
\end{abstract}

\section{Introduction}
\label{sec:introduction}

\red{Recently}, a collection of problems called \emph{Combinatorial Reconfiguration} has been extensively studied.
Work in this research area specifically aims to model dynamic situations where one needs to transform one feasible solution of a computational problem into another by locally changing a solution while keeping its feasibility along the way.
In a reconfiguration setting, two feasible solutions of a computational problem (e.g., \textsc{Satisfiability}, \textsc{Independent Set}, \textsc{Vertex Cover}, \textsc{Dominating Set}, etc.) are given, along with a \emph{reconfiguration rule} that describes an adjacency relation between solutions.
A \emph{reconfiguration problem} asks whether one feasible solution can be transformed into the other via a sequence of adjacent feasible solutions where each intermediate member is obtained from its predecessor by applying the given reconfiguration rule exactly once.
Such a sequence, if exists, is called a \emph{reconfiguration sequence}.
One may recall the classic Rubik's cube puzzle as an example of a reconfiguration problem, where each configuration of the Rubik's cube corresponds to a feasible solution, and two configurations (solutions) are adjacent if one can be obtained from the other by rotating a face of the cube by either $90$, $180$, or $270$ \red{degrees}.
The question is whether one can transform an arbitrary configuration to the one where each face of the cube has only one color.
For an overview of this research area, readers are referred to the recent surveys~\cite{Heuvel13,Nishimura18,MynhardtN19}.

\subsection{\textsc{$k$-Path Vertex Cover Reconfiguration}}

Let $G = (V, E)$ be a simple graph.
A \emph{vertex cover} of $G$ is a subset $I$ of $V$ where each edge contains at least one vertex from $I$.
The \textsc{Vertex Cover (VC)} problem, which asks whether there is a vertex cover of $G$ whose size is at most some positive integer $s$, is one of the classic $\mathtt{NP}$-complete problems in the computational complexity theory~\cite{GareyJohnson}.

Let $k \geq 2$ be a fixed positive integer.
A subset $I$ of $V$ is called a \emph{$k$-path vertex cover} if every path on $k$ vertices in $G$ contains at least one vertex from $I$.
The \textsc{$k$-Path Vertex Cover ($k$-PVC)} problem asks if there is a $k$-path vertex cover of $G$ whose size is at most some positive integer $s$.
Motivated by the importance of a problem related to secure communication in wireless sensor networks, Bre{\v{s}}ar et al. initiated the study of \textsc{$k$-PVC} in~\cite{BrevsarKKS11} (as a generalized concept of \emph{vertex cover}).
It is known that \textsc{$k$-PVC} is $\mathtt{NP}$-complete for every $k \geq 2$~\cite{AcharyaCBG12,BrevsarKKS11}.
Subsequent work regarding the \emph{maximum} variant~\cite{MiyanoSUYZ18} and \emph{weighted} variant~\cite{BrevsarKSS14} of \textsc{$k$-PVC} has also been considered in the literature.
Recently, the study of \textsc{$k$-PVC} and related problems has gained a lot of \red{attention} from both theoretical aspect~\cite{KumarKR19,RanZHLD19,Tsur19} 
and practical application~\cite{BeckLKSZ19,FunkeNS14}.

In this paper, we initiate the study of \textsc{$k$-PVC} from the viewpoint of \emph{combinatorial reconfiguration}.
Given two distinct $k$-path vertex covers $I$ and $J$ of a graph $G$ and a \emph{single} reconfiguration rule, the \textsc{$k$-Path Vertex Cover Reconfiguration ($k$-PVCR)} problem asks whether there is a reconfiguration sequence between $I$ and $J$.
We study the computational complexity of \textsc{$k$-PVCR} with respect to different graph classes under the well-known reconfiguration rules: \emph{Token Sliding}, \emph{Token Jumping}, and \emph{Token Addition or Removal}.
They are informally defined as follows.
Imagine that a token is placed at each vertex of a $k$-path vertex cover in $G$.
For each of the following rules, a common requirement is that the resulting token-set forms a $k$-path vertex cover of $G$.
\begin{itemize}
	\item Token Sliding ($\mathsf{TS}$): A $\mathsf{TS}$-step involves moving a token on some vertex $v$ to one of its unoccupied neighbors.
	\item Token Jumping ($\mathsf{TJ}$): A $\mathsf{TJ}$-step involves moving a token on $v$ to any unoccupied vertex.
	\item Token Addition or Removal ($\mathsf{TAR}$): A $\mathsf{TAR}$-step involves either adding or removing a single token such that the resulting token-set is of size at most given positive integer $u$. We sometimes write ``$\mathsf{TAR}(u)$'' instead of ``$\mathsf{TAR}$'' to emphasize the upper bound $u$ on the size of each token-set in a reconfiguration sequence under $\mathsf{TAR}$.
\end{itemize}

\subsection{Related Work}\label{subsection_related_work}

The \emph{reoptimization} framework is closely related to \emph{reconfiguration}. 
Roughly speaking, given an optimal solution of a problem instance $\mathcal{I}$, and some perturbations that change $\mathcal{I}$ into a new instance $\mathcal{I}^\prime$, a reoptimization problem aims to find an optimal solution for the changed instance $\mathcal{I}^\prime$.
Recently, Kumar~et~al.~\cite{KumarKR19} initiated the study of reoptimization problems for (both weighted and unweighted) \textsc{$k$-PVC} with $k \geq 3$, extending some known reoptimization paradigms for the well-known \textsc{VC} problem~\cite{AusielloBE11}.
The perturbation they considered in~\cite{KumarKR19} is changing the input graph of the current instance by inserting new vertices.

The \textsc{Vertex Cover Reconfiguration (VCR)} problem is one of the most well-studied reconfiguration problems, from both classical and parameterized complexity viewpoints (e.g., see~\cite{Nishimura18} for a quick summary of known results).
It is well-known that if $I$ is a vertex cover of a graph $G = (V, E)$ then $V \setminus I$ is an \emph{independent set} of $G$, i.e., a vertex subset whose members are pairwise non-adjacent.
Consequently, from classical complexity viewpoint, results of \textsc{Independent Set Reconfiguration (ISR)} and \textsc{Vertex Cover Reconfiguration} are interchangeable.

We now mention some known complexity results of \textsc{VCR} (which are mostly interchanged with \textsc{ISR}) for some graph classes.
It is well-known that \textsc{VCR} is $\mathtt{PSPACE}$-complete under each of $\mathsf{TS}$, $\mathsf{TJ}$, and $\mathsf{TAR}$ for general graphs~\cite{ItoDHPSUU11}, planar graphs of maximum degree three~\cite{HearnD05,ItoKOSUY20}, perfect graphs~\cite{KaminskiMM12}, and bounded bandwidth graphs~\cite{Wrochna18}.
Even on bipartite graphs, \textsc{VCR} remains $\mathtt{PSPACE}$-complete under $\mathsf{TS}$, and $\mathtt{NP}$-complete under each of $\mathsf{TJ}$ and $\mathsf{TAR}$~\cite{LokshtanovM19}.
On chordal graphs (and even on split graphs), \textsc{VCR} is known to be $\mathtt{PSPACE}$-complete under $\mathsf{TS}$~\cite{BelmonteKLMOS19}.
On the positive side, polynomial-time algorithms have been designed for \textsc{VCR} on even-hole-free graphs (and therefore chordal graphs) under each of $\mathsf{TJ}$ and $\mathsf{TAR}$~\cite{KaminskiMM12}, on bipartite permutation graphs and bipartite distance-hereditary graphs~\cite{Fox-EpsteinHOU15} under $\mathsf{TS}$, on cographs~\cite{Bonsma16,KaminskiMM12}, claw-free graphs~\cite{BonsmaKW14}, interval graphs~\cite{BonamyB17,KaminskiMM12,BrianskiFHM21}, and trees~\cite{DemaineDFHIOOUY15,KaminskiMM12} under each of $\mathsf{TS}$, $\mathsf{TJ}$, and $\mathsf{TAR}$.

\subsection{Our Results}

\begin{figure}[t]
	\centering
	\resizebox{\textwidth}{!}{
		\begin{tikzpicture}[level distance=2cm,
			level 1/.style={sibling distance=12em},
			level 2/.style={sibling distance=12em},
			every node/.style = {very thick, shape=rectangle, rounded corners, draw, align=center, minimum height=2.7em},
			edge from parent/.style={draw, thick, -stealth}]
			
			\node (general) {general}
			child {node (chordal) {chordal} }
				child {node [shift={(0,-4)}] (tree) {tree} }
			child {node (planar) {planar} 
				child {node [shift={(2,0)}] (planar-bounded-bw) {planar and\\ bounded bandwidth} 
					child { node (subcubic-planar-bounded-bw) {planar, maximum degree $3$,\\ and bounded bandwidth}
						child { node (path) {path} }
						child { node (cycle) {cycle} }
					} 
				}
			}
			child {node (bounded-bw) {bounded bandwidth} }
			;
			
			\draw[thick, -stealth] (bounded-bw) -- (planar-bounded-bw);
			\draw[thick, -stealth] (planar-bounded-bw) -- (tree);
			\draw[thick, -stealth] (chordal) -- (tree);
			\draw[thick, -stealth] (tree) -- (path);
			
			\draw [line width=0.1cm, black] ([shift={(1,0.2)}]cycle.north) -- ([shift={(0,-4.85)}]chordal.south west) node [below, draw=none, fill=none, shift={(1,0)}] {$\mathsf{TS}$\\ Poly-time};
			\draw [line width=0.1cm, black] ([shift={(0.5,-0.8)}]tree.south east) -- ([shift={(0.5,0.3)}]tree.north east) -- ([shift={(0,-2.7)}]chordal.south west) node [below, draw=none, fill=none, shift={(1,0)}] {$\mathsf{TS}$\\ Unknown} node [above, draw=none, fill=none, shift={(1,0)}] {$\mathsf{TS}$\\ $\mathtt{PSPACE}$-complete};
			\draw [line width=0.1cm, black] ([shift={(1,0.4)}]cycle.north) -- ([shift={(0.7,-0.6)}]tree.south east) -- ([shift={(0.7,2.4)}]tree.north east) -- ([shift={(0,-0.6)}]chordal.south west) node [below, draw=none, fill=none, shift={(1,0)}] {$\mathsf{TJ}/\mathsf{TAR}(u)$\\ Poly-time};
			\draw [line width=0.1cm, black] ([shift={(0.8,-0.8)}]chordal.south east) -- ([shift={(0.8,1)}]chordal.north east) -- ([shift={(0,1)}]chordal.north west) node [below, draw=none, fill=none, shift={(1,0)}] {$\mathsf{TJ}/\mathsf{TAR}(u)$\\ Unknown};
			\draw [line width=0.1cm, black] ([shift={(0.7,2.4)}]tree.north east) -- ([shift={(-0.1,0)}]general.north west) -- ([shift={(0,2)}]chordal.north west) 
			node [above, draw=none, fill=none, shift={(1,0)}] {$\mathsf{TJ}/\mathsf{TAR}(u)$\\ $\mathtt{PSPACE}$-complete};
			
			\node [draw=none, inner sep=0, shift={(0.8, 0.3)}, below of=planar-bounded-bw] {[Thm.~\ref{thm:PSPACEC}]};
			\node [draw=none, inner sep=0, shift={(0, 0.3)}, below of=chordal] {[Thm.~\ref{thm:PSPACEC}]};
			\node [draw=none, inner sep=0, shift={(1.2, 0.3)}, below of=subcubic-planar-bounded-bw] {[Thm.~\ref{thm:PSPACEC-planar-maxdeg3}]};
			\node [draw=none, inner sep=0, shift={(-0.2, 0.3)}, below of=tree] {[Thm.~\ref{thm:kPVC-Trees-TJ} and \ref{thm:kPVCR-Trees-TAR}]};
			\node [draw=none, inner sep=0, shift={(-0.2, 0.3)}, below of=path] {[Thm.~\ref{thm:kPVCR_paths_TJ}, \ref{thm:kPVCR_paths_TAR}, and \ref{thm:kPVCR_paths_TS}]};
			\node [draw=none, inner sep=0, shift={(-0.2, 0.3)}, below of=cycle] {[Thm.~\ref{thm:kPVCR_cycle_TJTS} and \ref{thm:kPVCR_cycles_TAR}]};
		\end{tikzpicture}
	}
	\caption{\red{Computational complexity of \textsc{$k$-PVCR} on some graph classes, under each of $\mathsf{TS}$, $\mathsf{TJ}$, and $\mathsf{TAR}(u)$. Each arrow from $A$ to $B$ means $B$ is a subclass of $A$.}}
	\label{fig:results}
\end{figure}

In this paper, we investigate the complexity of \textsc{$k$-PVCR} with respect to different input graphs (see \figurename~\ref{fig:results}).
More precisely, we show that:
\begin{itemize}
	\item Several hardness results for \textsc{VCR} remain true for \textsc{$k$-PVCR}. 
	More precisely, we show the $\mathtt{PSPACE}$-completeness of \textsc{$k$-PVCR} on general graphs under each rule $\mathsf{TS}$, $\mathsf{TJ}$, and $\mathsf{TAR}$ using a reduction from a variant of \textsc{VCR}.
	As our reduction preserves some nice graph properties, we claim (as a consequence of our reduction) that the hardness results for \textsc{VCR} on several graphs (namely planar graphs, bounded bandwidth graphs, chordal graphs)
	can be converted into those for \textsc{$k$-PVCR}.
	Using a reduction from the \textsc{Nondeterministic Constraint Logic}~\cite{HearnD05,Z15}, we also show that \textsc{$k$-PVCR} remains $\mathtt{PSPACE}$-complete even on planar graphs of bounded bandwidth and maximum degree three.
	(Our reduction from \textsc{VCR} does not preserve the maximum degree.) 
	
	\item On the positive side, we design polynomial-time algorithms for \textsc{$k$-PVCR} on some restricted graph classes: trees (under each of $\mathsf{TJ}$ and $\mathsf{TAR}$), paths and cycles (under each of $\mathsf{TS}$, $\mathsf{TJ}$, and $\mathsf{TAR}$).
	Our algorithms are constructive, i.e., we explicitly show how a reconfiguration sequence can be constructed in a yes-instance.
	On paths, we claim that our algorithm constructs a \emph{shortest} reconfiguration sequence.
	As a result, we obtain a complexity dichotomy for \textsc{$k$-PVCR} on (planar) graphs with respect to their maximum degree.
\end{itemize}

\section{Preliminaries}\label{sec:preliminaries}

In this section, we define some useful notation and terminology.
For standard concepts on graphs, we refer readers to~\cite{Diestel2010}.

Let $G$ be a simple graph with vertex-set $V(G)$ and edge-set $E(G)$.
For two vertices $u, v$, we denote by $\text{dist}_G(u, v)$ the \emph{distance} between $u$ and $v$ in $G$, i.e., the number of edges in a shortest path between them.
For a vertex $v \in V(G)$, we denote by $G - v$ the graph obtained from $G$ by removing the vertex $v$ and all incident edges.
For two vertex subsets $I$ and $J$, we denote by $G[I \Delta J]$ the subgraph of $G$ induced by their \emph{symmetric difference} $I \Delta J = (I \setminus J) \cup (J \setminus I)$.
For a fixed integer $k \geq 2$, we say that a vertex $v$ \emph{covers} a $k$-path (i.e., a path on $k$ vertices) $P_k$ in $G$ if $v \in V(P_k)$.
A vertex subset $I$ is called a \emph{$k$-path vertex cover} if every $k$-path in $G$ contains at least one vertex from $I$.
In other words, vertices of $I$ cover all $k$-paths in $G$.
We denote by $\psi_k(G)$ the size of a \emph{minimum} $k$-path vertex cover of $G$.
Trivially, for $n \geq k \geq 2$, $\psi_k(P_n) = \lfloor n/k \rfloor$ and $\psi_k(C_n) = \lceil n/k \rceil$ for a path $P_n$ and a cycle $C_n$ on $n$ vertices.

Throughout this paper, we denote by $(G, I, J, \mathsf{R})$ an instance of \textsc{$k$-PVCR} under a reconfiguration rule $\mathsf{R}\in \{\mathsf{TJ}, \mathsf{TS}, \mathsf{TAR}\}$, where $I$ and $J$ are two $k$-path vertex covers of $G$. 
We shall  respectively call a reconfiguration sequence under each of $\mathsf{TS}$, $\mathsf{TJ}$, and $\mathsf{TAR}$ by a \emph{$\mathsf{TS}$-sequence}, \emph{$\mathsf{TJ}$-sequence}, and \emph{$\mathsf{TAR}(u)$-sequence}.
Formally, let $S = \langle I_0, I_1, \dots, I_\ell \rangle$ be an \emph{ordered} sequence of $k$-path vertex covers of $G$.
The \emph{length} of $S$ is defined as $\ell$, i.e., if $S$ is a reconfiguration sequence then its length is exactly the number of steps it performs under the given reconfiguration rule.
Imagine that a token is placed at each vertex of a $k$-path vertex cover of $G$.
We may sometimes identify a token with the vertex where it is placed and say ``a token in a $k$-path vertex cover'', and therefore use the terms ``token-set'' and ``$k$-path vertex cover'' interchangeably.
We say that $S$ is a \emph{$\mathsf{TS}$-sequence} between two $k$-path vertex covers $I_0$ and $I_\ell$ if for each $i \in \{0, \dots, \ell-1\}$, there exist two vertices $x_i$ and $y_i$ such that $I_i \setminus I_{i+1} = \{x_i\}$, $I_{i+1} \setminus I_i = \{y_i\}$, and $x_iy_i \in E(G)$.
Roughly speaking, $I_{i+1}$ is obtained from $I_i$ by sliding the token placed on $x_i$ to $y_i$ along an edge $x_iy_i$.
Similarly, we say that $S$ is a \emph{$\mathsf{TJ}$-sequence} between $I_0$ and $I_\ell$ if for each $i \in \{0, \dots, \ell-1\}$, there exist two vertices $x_i$ and $y_i$ such that $I_i \setminus I_{i+1} = \{x_i\}$, $I_{i+1} \setminus I_i = \{y_i\}$.
Intuitively, $I_{i+1}$ is obtained from $I_i$ by jumping the token placed on $x_i$ to $y_i$.
Now, if $\max\{|I_i|: 0 \leq i \leq \ell\} \leq u$ for some positive integer $u$, and for each $i \in \{0, \dots, \ell-1\}$, there exists a vertex $x_i$ such that $I_i \Delta I_{i+1} = \{x_i\}$ then we say that $S$ is a $\mathsf{TAR}(u)$-sequence between $I_0$ and $I_\ell$.
Roughly speaking, $I_{i+1}$ is obtained from $I_i$ by either adding a token to $x_i$ or removing a token from $x_i$.
If a $\mathsf{TS}$-, $\mathsf{TJ}$-, or $\mathsf{TAR}(u)$-sequence between two $k$-path vertex covers $I$ and $J$ exists, we say that $I$ and $J$ are \emph{reconfigurable} under $\mathsf{TS}$, $\mathsf{TJ}$, or $\mathsf{TAR}$, respectively.

Using a similar argument as in~\cite[Theorem~1]{KaminskiMM12}, we can prove the following useful lemma.
\begin{lemma}
	\label{lem:TJ-TAR-equiv}
	There exists a $\mathsf{TJ}$-sequence of length $\ell$ between two $k$-path vertex covers $I, J$ of a graph $G$ with $|I| = |J| = s$ if and only if there exists a $\mathsf{TAR}(s+1)$-sequence of length $2\ell$ between them.
\end{lemma}
%
%

A reconfiguration sequence of minimum length is called a \emph{shortest} reconfiguration sequence.
For a reconfiguration sequence $S = \langle I_0, I_1, \dots, I_p \rangle$, we denote by $\text{rev}(S)$ the \emph{reverse} of $S$, i.e., the reconfiguration sequence $\langle I_p, \dots, I_1, I_0\rangle$.
For two reconfiguration sequences $S = \langle I_0, I_1, \dots, I_p \rangle$ and $S^\prime = \langle I^\prime_0, I^\prime_1, \dots, I^\prime_q \rangle$ under the same reconfiguration rule, if $I_p = I^\prime_0$ then we say that they can be \emph{concatenated} and define their \emph{concatenation} $S \oplus S^\prime$ as the reconfiguration sequence $\langle I_0, I_1, \dots, I_p, I^\prime_1, \dots, I^\prime_q \rangle$.
We assume for convenience that if $S^\prime$ is empty then $S \oplus S^\prime = S^\prime \oplus S = S$.

\section{Hardness Results}
\label{sec:hardness}

\subsection{Reduction from \textsc{Vertex Cover Reconfiguration}}

In this section, we prove the following theorem using a polynomial-time reduction from \textsc{VCR}.

\begin{theorem}
	\label{thm:PSPACEC}
	\textsc{$k$-PVCR} is $\mathtt{PSPACE}$-complete under each of $\mathsf{TS}$, $\mathsf{TJ}$, and $\mathsf{TAR}$ even when the input graph is a planar graph of maximum degree four, or a bounded bandwidth graph.
	Additionally, \textsc{$k$-PVCR} is $\mathtt{PSPACE}$-complete under $\mathsf{TS}$ on chordal graphs.
\end{theorem}

The outline of our proof is as follows: 
\begin{itemize}
	\item[(1)] In Lemma~\ref{lem:PSPACEC-TJ}, using a \red{reduction similar to that} in~\cite{BrevsarKKS11}, we show the $\mathtt{PSPACE}$-completeness of \textsc{$k$-PVCR} under $\mathsf{TJ}$. 
	\item[(2)] In Lemma~\ref{lem:PSPACEC-TJ-TAR}, we combine (1), the known results for \textsc{VCR}, and Lemma~\ref{lem:TJ-TAR-equiv} to show the hardness results on several graphs under each of $\mathsf{TJ}$ and $\mathsf{TAR}$ as mentioned in Theorem~\ref{thm:PSPACEC}. 
	\item[(3)] Finally, in Lemma~\ref{lem:PSPACEC-TS}, we show that the hardness results under $\mathsf{TS}$ hold via the same reduction. 
\end{itemize}


\begin{lemma}
	\label{lem:PSPACEC-TJ}
	\textsc{$k$-PVCR} is $\mathtt{PSPACE}$-complete under $\mathsf{TJ}$. 
\end{lemma}
\begin{proof}
	Given two distinct \emph{minimum} $k$-path vertex covers $I$ and $J$ of a graph $G$ and a single reconfiguration rule, the \textsc{Minimum $k$-Path Vertex Cover Reconfiguration (Min-$k$-PVCR)} problem asks whether there is a reconfiguration sequence between $I$ and $J$. For $k=2$, the \textsc{Min-$k$-PVCR} problem is also known as \textsc{Minimum Vertex Cover Reconfiguration (Min-VCR)}. 
	
	Clearly, since \textsc{$k$-Path Vertex Cover} is in $\mathtt{NP}$~\cite{BrevsarKKS11}, it follows from~\cite{ItoDHPSUU11} that \textsc{$k$-PVCR} is in $\mathtt{PSPACE}$.
	Since \textsc{$k$-PVCR} is more general than \textsc{Min-$k$-PVCR}, in order to show the $\mathtt{PSPACE}$-completeness of \textsc{$k$-PVCR}, it suffices to reduce from the \textsc{Min-VCR} problem (which is known to be $\mathtt{PSPACE}$-complete~\cite{ItoDHPSUU11}) to the \textsc{Min-$k$-PVCR} problem.
	More precisely, given an instance $(G, I, J, \mathsf{TJ})$ of \textsc{Min-VCR}, we construct a corresponding instance $(G^\prime, I^\prime, J^\prime, \mathsf{TJ})$ of \textsc{Min-$k$-PVCR} as follows.
	Let $G^\prime$ be the graph obtained from $G$ by joining each vertex $x$ of $G$ to an endpoint of a new path $P^x$ on $\lfloor (k-1)/2\rfloor$ vertices.
	We choose $I^\prime = I$ and $J^\prime = J$.
	Note that each vertex cover of $G$ is also a $k$-path vertex cover of $G^\prime$, 
	Moreover, for any minimum $k$-path vertex cover $I^\prime$ of $G^\prime$, if $I^\prime$ contains a new vertex $y$ in a path $P^x$ for some vertex $x$ of $G$ then $(I^\prime \setminus \{y\}) \cup \{x\}$ is also a minimum $k$-path vertex cover of $G^\prime$, because any $k$-path covered by $y$ must also be covered by $x$.
	Consequently, $(G^\prime, I^\prime, J^\prime, \mathsf{TJ})$ is an instance of \textsc{Min-$k$-PVCR}.
	
	It is clear that this construction can be done in polynomial time.
	It remains to show that $(G, I, J, \mathsf{TJ})$ is a yes-instance of \textsc{Min-VCR} if and only if $(G^\prime, I^\prime, J^\prime, \mathsf{TJ})$ is a yes-instance of \textsc{Min-$k$-PVCR}.
	
	Assume that $(G, I, J, \mathsf{TJ})$ is a yes-instance of \textsc{Min-VCR}, that is, there exists a $\mathsf{TJ}$-sequence $\langle I = I_0, I_1, \dots, I_p = J \rangle$ between $I$ and $J$ in $G$.
	Clearly, for any $i \in \{0, 1, \dots, p\}$, the set $I_i$ is also a minimum $k$-path vertex cover of $G^\prime$.
	Then, $\langle I = I_0, I_1, \dots, I_p = J \rangle$ is also a $\mathsf{TJ}$-sequence between $I^\prime = I$ and $J^\prime = J$ in $G^\prime$.
	
	Now, assume that $(G^\prime, I^\prime, J^\prime, \mathsf{TJ})$ is a yes-instance of \textsc{Min-$k$-PVCR} in $G^\prime$, that is, there exists a $\mathsf{TJ}$-sequence $S = \langle I^\prime = I^\prime_0, I^\prime_1, \dots, I^\prime_q = J^\prime \rangle$ between $I^\prime = I$ and $J^\prime = J$ in $G^\prime$.
	We claim that $(G, I, J, \mathsf{TJ})$ is also a yes-instance by constructing a $\mathsf{TJ}$-sequence between $I$ and $J$ in $G$. 
	For $i \in \{0, 1, \dots, q\}$, let $I_i = I^\prime_i \setminus \bigcup_{x \in V(G)}V(P^x) \cap \bigcup_{x \in V(G)}\{x: I^\prime_i \cap V(P^x) \neq \emptyset\}$.
	Intuitively, $I_i$ is obtained from $I^\prime_i$ by moving each token placed at some new vertex in $P^x$ to $x$ itself.
	Since any $k$-path covered by some vertex in $P^x$ is also covered by $x$, and each $I^\prime_i$ is minimum, such moves are well-defined.
	Clearly, each $I_i$ is a minimum vertex cover of $G$.
	For $i \in \{0, 1, \dots, q-1\}$, let $x^\prime_i$ and $y^\prime_i$ be two distinct vertices of $G^\prime$ such that $I^\prime_i \setminus I^\prime_{i+1} = \{x^\prime_i\}$ and $I^\prime_{i+1} \setminus I^\prime_i = \{y^\prime_i\}$.
	Next, we will show that $I_{i+1}$ can be obtained from $I_i$ by performing at most one $\mathsf{TJ}$-step in $G$.
	\begin{itemize}
		\item {\bf Case~1: $x^\prime_i \in V(G)$ and $y^\prime_i \in V(G)$}.
		By definition, $I_i \setminus I_{i+1} = \{x^\prime_i\}$ and $I_{i+1} \setminus I_i = \{y^\prime_i\}$.
		\item {\bf Case~2: $x^\prime_i \in V(G)$ and $y^\prime_i \in V(G^\prime) \setminus V(G)$.}
		Then, $y^\prime_i$ must belong to a new path $P^y$ joined to some vertex $y \in V(G)$.
		By definition, $I_i \setminus I_{i+1} = \{x^\prime_i\}$ and $I_{i+1} \setminus I_i = \{y\}$.
		Note that if $x^\prime_i = y$, then $I_i = I_{i+1}$, and we are done.
		Moreover, as we consider minimum $k$-path vertex covers, $y \notin I^\prime_{i} \setminus \{x^\prime_i\}$ and therefore $y \notin I_{i}$; otherwise, we cannot move the token on $x^\prime_i$ to $y^\prime_i$.
		\item {\bf Case~3: $x^\prime_i \in V(G^\prime) \setminus V(G)$ and $y^\prime_i \in V(G)$.}
		As before, $x^\prime_i$ must belong to a new path $P^x$ joined to some vertex $x \in V(G)$.
		By definition, $I_i \setminus I_{i+1} = \{x\}$ and $I_{i+1} \setminus I_i = \{y^\prime_i\}$.
		Note that if $x = y^\prime_i$, then $I_i = I_{i+1}$.
		\item {\bf Case~4: $x^\prime_i \in V(G^\prime) \setminus V(G)$ and $y^\prime_i \in V(G^\prime) \setminus V(G)$.}
		As before, $x^\prime_i$ (resp. $y^\prime_i$) must belong to a new path $P^x$ (resp. $P^y$) joined to some vertex $x \in V(G)$ (resp. $y \in V(G)$).
		By definition, $I_i \setminus I_{i+1} = \{x\}$ and $I_{i+1} \setminus I_i = \{y\}$.
		Note that if $x = y$, then $I_i = I_{i+1}$.
	\end{itemize}
	Clearly, the sequence obtained from $\langle I_0, I_1, \dots, I_q \rangle$ by removing redundant vertex covers (i.e., those equal to their predecessors) is a $\mathsf{TJ}$-sequence in $G$ that reconfigures $I = I_0$ to $J = I_q$.
\end{proof}

\begin{lemma}
	\label{lem:PSPACEC-TJ-TAR}
	\textsc{$k$-PVCR} is $\mathtt{PSPACE}$-complete under each of $\mathsf{TJ}$ and $\mathsf{TAR}$ on planar graphs of maximum degree four and bounded bandwidth.
\end{lemma}
\begin{proof}
	\red{
	As we mention in Section~\ref{subsection_related_work}, it is known that \textsc{VCR} is $\mathtt{PSPACE}$-complete under each of $\mathsf{TJ}$, and $\mathsf{TAR}$ for planar graphs of maximum degree three~\cite{ItoKOSUY20}, and bounded bandwidth graphs~\cite{Wrochna18}.
	In fact, these results are also hold in the case MIN-VCR~\cite{ItoKOSUY20,Wrochna18}.
	}
	It is not hard to see that in the reduction presented in the proof of Lemma~\ref{lem:PSPACEC-TJ}, if the input graph $G$ is one of the mentioned graphs, then so is the constructed graph $G^\prime$.
	(In fact the bandwidth of $G^\prime$ is $O(k)$. However, since we defined that $k$ is a fixed integer, $G^\prime$ is of bounded bandwidth.)
	The hardness results under $\mathsf{TAR}$ are followed by combining the known results for \textsc{Vertex Cover Reconfiguration}, the above results, and Lemma~\ref{lem:TJ-TAR-equiv}.
\end{proof}

\begin{lemma}
	\label{lem:PSPACEC-TS}
	\textsc{$k$-PVCR} is $\mathtt{PSPACE}$-complete under $\mathsf{TS}$ on planar graphs of maximum degree \red{four} and bounded bandwidth, and chordal graphs.
\end{lemma}
\begin{proof}
	It is not hard to see that in the reduction presented in the proof of Lemma~\ref{lem:PSPACEC-TJ}, if the input graph $G$ is one of the mentioned graphs, then so is the constructed graph $G^\prime$.
	
	It is sufficient to show that any $\mathsf{TJ}$-sequence $S = \langle I_0, I_1, \dots, I_q \rangle$ between two minimum $k$-path vertex covers $I = I_0$ and $J = I_q$ of the constructed graph $G^\prime$ can be converted into a $\mathsf{TS}$-sequence between them in $G^\prime$.
	
	First of all, if $I_i \subseteq V(G)$ for all $i \in \{0, 1, \dots, q\}$ then we claim that $S$ itself is indeed a $\mathsf{TS}$-sequence.
	More precisely, we show that for each $i \in \{0, 1, \dots, q-1\}$, if $x_i$ and $y_i$ are two distinct vertices of $G$ such that $I_i \setminus I_{i+1} = \{x_i\}$ and $I_{i+1} \setminus I_i = \{y_i\}$ then $x_iy_i \in E(G) \subseteq E(G^\prime)$.
	Suppose to the contrary that $y_i$ is not adjacent to $x_i$.
	We note that each $I_i$ ($i \in \{0, 1, \dots, q\}$) is also a minimum vertex cover of $G$.
	Now, in order to move the token on $x_i$ to $y_i$ for obtaining a new vertex cover $I_{i+1}$ of $G$, each edge of $G$ incident with $x_i$ must already be covered by its other endpoint; otherwise, moving $x_i$ to $y_i$ left some non-covered edge.
	However, this means that one can obtain a vertex cover of smaller size by simply removing $x_i$ from $I_i$, which contradicts the fact that $I_i$ is minimum.
	Therefore, $y_i$ must be a neighbor of $x_i$.
	
	Now, from the above reduction, we know that there is always a $\mathsf{TJ}$-sequence $S^\prime$ between two $k$-path vertex covers $I^\prime = I \setminus \bigcup_{x \in V(G)}V(P^x) \cap \bigcup_{x \in V(G)}\{x: I \cap V(P^x) \neq \emptyset\}$ and $J^\prime = J \setminus \bigcup_{x \in V(G)}V(P^x) \cap \bigcup_{x \in V(G)}\{x: J \cap V(P^x) \neq \emptyset\}$, where all members of $S^\prime$ are subsets of $V(G)$.
	Here $P^x$ denotes the new path joined to the vertex $x \in V(G)$.
	As a result, $S^\prime$ is also a $\mathsf{TS}$-sequence in $G^\prime$.
	To construct a $\mathsf{TS}$-sequence between $I$ and $J$, it suffices to show that one can construct a $\mathsf{TS}$-sequence $S^{\prime\prime}$ between $I$ and $I^\prime$ in $G^\prime$.
	In a similar manner, we will be able to construct a $\mathsf{TS}$-sequence between $J$ and $J^\prime$, and a $\mathsf{TS}$-sequence between $I$ and $J$ can be formed by simply reconfiguring $I$ to $I^\prime$, then $I^\prime$ to $J^\prime$, and finally $J^\prime$ to $J$.
	Let $x \in V(G)$ be such that $I \cap V(P^x) = \{x^\prime\}$.
	Since $I$ is a minimum $k$-path vertex cover of $G^\prime$, we have $x \notin I$.
	We claim that $I$ can be reconfigured to $I \setminus \{x^\prime\} \cup \{x\}$ using $\mathsf{TS}$-steps.
	Let $P = v_0v_1\dots v_\ell$ ($0 \leq \ell \leq \lfloor (k-1)/2 \rfloor$) be the unique path in $G^\prime$ joining $v_0 = x$ and $v_\ell = x^\prime$.
	Note that for each $j \in \{1, \dots, \ell\}$, any $k$-path covered by $v_j$ is also covered by each vertex in $\{v_0, \dots, v_{j-1}\}$.
	Moreover, as we consider minimum $k$-path vertex covers, exactly one of $v_j$ ($j \in \{0, 1, \dots, \ell\}$) contains a token.
	Hence, one can obtain $I \setminus \{x^\prime\} \cup \{x\}$ from $I$ by simply sliding the token on $x^\prime \in I$ to $x$ along the path $P$.
	Applying this process repeatedly for each $x \in V(G)$ where $I \cap V(P^x) \neq \emptyset$, we obtain a $\mathsf{TS}$-sequence in $G^\prime$ between $I$ and $I^\prime$.
\end{proof}

Our proof of Theorem~\ref{thm:PSPACEC} is complete.

\subsection{Reduction from \textsc{Nondeterministic Constraint Logic}}

In Theorem~\ref{thm:PSPACEC}, we show the $\mathtt{PSPACE}$-completeness for planar graphs of maximum degree \red{four}.
Furthermore, using a reduction from the \textsc{Nondeterministic Constraint Logic}~\cite{HearnD05,Z15} (NCL, for short), we can improve this result as follows. 

\begin{theorem}
	\label{thm:PSPACEC-planar-maxdeg3}
	\textsc{$k$-PVCR} remains $\mathtt{PSPACE}$-complete under each of $\mathsf{TS}$, $\mathsf{TJ}$, and $\mathsf{TAR}$ even on planar graphs of bounded bandwidth and maximum degree \red{three}.
\end{theorem}

In this section, we briefly define NCL and show the cases for $\mathsf{TS}$ and $\mathsf{TJ}$, because the case for $\mathsf{TAR}$ can be shown similar to the proof of Lemma~\ref{lem:TJ-TAR-equiv}.
This result can be obtained by constructing polynomial-time reductions from NCL---a well-known $\mathtt{PSPACE}$-complete problem first introduced by Hearn and Demaine~\cite{HearnD05}.
This problem is often used to prove the computational hardness of puzzles and games, because a reduction from this problem requires to construct only two types of gadgets, called \textsc{and} and \textsc{or} gadgets.

\subsubsection{\textsc{Nondeterministic Constraint Logic}} 
\begin{figure}[b]
	\begin{minipage}{0.33\hsize}
		\centering
		\scalebox{0.7}{
			\begin{tikzpicture}
				\node (v1) at (1,3) [circle, draw] {};
				\node (v2) at (5,3) [circle, draw] {};
				\node (v3) at (2,2) [circle, draw] {};
				\node (v4) at (4,2) [circle, draw] {};
				\node (v5) at (1,1) [circle, draw] {};
				\node (v6) at (5,1) [circle, draw] {};
				\draw (v1) -- node[midway, above left]{2} (v2) [ultra thick, blue] {};
				\draw (v1) -- node[midway, above right]{2} (v3) [ultra thick, blue] {};
				\draw (v1) -- node[midway, above left]{2} (v5) [ultra thick, blue] {};
				\draw (v3) -- node[midway, above left]{2} (v4) [ultra thick, blue] {};
				\draw (v3) -- node[midway, right]{2} (v5) [ultra thick, blue] {};
				\draw (v5) -- node[midway, above left]{2} (v6) [ultra thick, blue] {};
				\draw (v2) -- node[midway, left]{1} (v4) [red] {};
				\draw (v2) -- node[midway, above right]{1} (v6) [red] {};
				\draw (v4) -- node[midway, left]{1} (v6) [red] {};
				\draw (3.1,3.1) -- (2.9,3.0) -- (3.1,2.9) [ultra thick] {};
				\draw (2.9,2.1) -- (3.1,2.0) -- (2.9,1.9) [ultra thick] {};
				\draw (2.9,1.1) -- (3.1,1.0) -- (2.9,0.9) [ultra thick] {};
				\draw (0.9,2.1) -- (1.0,1.9) -- (1.1,2.1) [ultra thick] {};
				\draw (4.9,1.9) -- (5.0,2.1) -- (5.1,1.9) [] {};
				\draw (1.58,2.58) -- (1.65,2.35) -- (1.42,2.42) [ultra thick] {};
				\draw (1.42,1.58) -- (1.35,1.35) -- (1.58,1.42) [ultra thick] {};
				\draw (4.42,2.58) -- (4.65,2.65) -- (4.58,2.42) [] {};
				\draw (4.42,1.42) -- (4.65,1.35) -- (4.58,1.58) [] {};
			\end{tikzpicture}
		}
		
		(a)
	\end{minipage}
	\begin{minipage}{0.33\hsize}
		\centering
		\scalebox{1.0}{
			\begin{tikzpicture}
				\node (v0) at (1.0,1.0) [circle, draw] {};
				\node (v1) at (0.2,0.5) [label=above:1] {};
				\node (v2) at (1.8,0.5) [label=above:1] {};
				\node (v3) at (1.0,2.0) [label=below right:2] {};
				\draw (v0) -- (v1) [red] {};
				\draw (v0) -- (v2) [red] {};
				\draw (v0) -- (v3) [ultra thick, blue] {};
			\end{tikzpicture}
		}
		
		(b)
	\end{minipage}
	\begin{minipage}{0.33\hsize}
		\centering
		\scalebox{1.0}{
			\begin{tikzpicture}
				\node (v0) at (1.0,1.0) [circle, draw] {};
				\node (v1) at (0.2,0.5) [label=above:2] {};
				\node (v2) at (1.8,0.5) [label=above:2] {};
				\node (v3) at (1.0,2.0) [label=below right:2] {};
				\draw (v0) -- (v1) [ultra thick, blue] {};
				\draw (v0) -- (v2) [ultra thick, blue] {};
				\draw (v0) -- (v3) [ultra thick, blue] {};
			\end{tikzpicture}
		}
		
		(c)
	\end{minipage}
	\caption{(a) A configuration of an NCL machine, (b) NCL \textsc{and} vertex, and (c) NCL \textsc{or} vertex.}
	\label{fig:NCL}
\end{figure}
Now we define NCL problem~\cite{HearnD05}.
An NCL ``machine'' is an undirected graph together with an assignment of weights from $\{1,2\}$ to each edge of the graph. 
An (\emph{NCL}) \emph{configuration} of this machine is an orientation (direction) of the edges such that the sum of weights of in-coming arcs at each vertex is at least two. 
\figurename~\ref{fig:NCL}(a) illustrates a configuration of an NCL machine, where each weight-2 edge is depicted by a thick (blue) line and each weight-1 edge by a thin (red) line. 
Then, two NCL configurations are \emph{adjacent} if they differ in a single edge direction. 
Given an NCL machine and its two configurations, it is known to be $\mathtt{PSPACE}$-complete to determine whether there exists a sequence of adjacent NCL configurations which transforms one into the other~\cite{HearnD05}.     

An NCL machine is called an \emph{\textsc{and}/\textsc{or} constraint graph} if it consists of only two types of vertices, called ``NCL \textsc{and} vertices'' and ``NCL \textsc{or} vertices'' defined as follows:
\begin{itemize}
	\item A vertex of degree three is called an \emph{NCL \textsc{and} vertex} if its three incident edges have weights $1$, $1$ and $2$. 
	(See \figurename~\ref{fig:NCL}(b).)
	An NCL \textsc{and} vertex $u$ behaves as a logical \textsc{and}, in the following sense: 
	the weight-$2$ edge can be directed outward for $u$ if and only if both two weight-$1$ edges are directed inward for $u$. 
	Note that, however, the weight-$2$ edge is not necessarily directed outward even when both weight-$1$ edges are directed inward. 
	\item A vertex of degree three is called an \emph{NCL \textsc{or} vertex} if its three incident edges have weights $2$, $2$ and $2$. 
	(See \figurename~\ref{fig:NCL}(c).)
	An NCL \textsc{or} vertex $v$ behaves as a logical \textsc{or}: 
	one of the three edges can be directed outward for $v$ if and only if at least one of the other two edges is directed inward for $v$. 
\end{itemize}
It should be noted that, although it is natural to think of NCL \textsc{and}/\textsc{or} vertices as having inputs and outputs, there is nothing enforcing this interpretation; 
especially for NCL \textsc{or} vertices, the choice of input and output is entirely arbitrary because an NCL \textsc{or} vertex is symmetric. 

For example, the NCL machine in \figurename~\ref{fig:NCL}(a) is an \textsc{and}/\textsc{or} constraint graph. 
From now on, we call an \textsc{and}/\textsc{or} constraint graph simply an \emph{NCL machine}, and call an edge in an NCL machine an \emph{NCL edge}. 
NCL remains $\mathtt{PSPACE}$-complete even if an input NCL machine is planar and bounded bandwidth~\cite{Z15}. 

\subsubsection{Constructing gadgets} 

In our reduction, we construct two types of gadgets named \emph{\textsc{and}/\textsc{or} gadgets}, which correspond to NCL \textsc{and}/\textsc{or} vertices, respectively.
Both \textsc{and}/\textsc{or} gadgets consist of one \emph{main part} and three \emph{connecting parts}.
Each connecting part corresponds to each incident NCL edge of the corresponding vertex.
Then we replace each of vertices in \red{the} NCL machine with its corresponding gadget so that each pair of adjacent vertices sharing their connecting parts.

Each connecting part is formed $P_{2k-2}$.
Note that if we want to cover this path with only one vertex, we must choose one of the two center vertices.
In our reduction, choosing one of the two vertices corresponds to inward direction, and the other one corresponds to outward direction.

Now we explain the construction of \red{the} \textsc{and} gadget.
Consider an NCL \textsc{and} vertex.
\figurename~\ref{figure:and_reconf}(a) illustrates all valid orientations of the edges incident to an NCL \textsc{and} vertex.
Two boxes are joined by an edge if their \red{orientations} are adjacent.
We construct our \textsc{and} gadget so that it correctly simulates this reconfiguration graph in Fig.~\ref{figure:and_reconf}(a).

\figurename~\ref{figure:NCLgadget}(a) illustrates our \textsc{and} gadget for the case where $k=3$.
The main part of \textsc{and} gadget forms $P_k$.
Note that we must choose at least one of the vertices on this part to obtain $k$-PVC.
Then we connect one endpoint to two connecting parts which \red{corresponds to} weight-1 edges, and connect the other endpoint to a connecting part which \red{corresponds to} weight-2 edge.
\red{
If at least one of the weight-1 edges is directed outward, we must choose the endpoint of main part next to the connecting part corresponding to the weight-1 edge to obtain $k$-PVC.
On the other hand, if the weight-2 edge is directed outward, we must choose the endpoint of main part next to the connecting part corresponding to the weight-2 edge to obtain $k$-PVC.
}
Fig.~\ref{figure:and_reconf}(b) illustrates the reconfiguration graph for all $3$-PVCs of the \textsc{and} gadget where we allow to choose at most four vertices as $3$-PVC.
Each large dashed box surrounds all $3$-PVCs choosing the same vertices from their connecting part.
Then we can see that these $3$-PVCs are ``internally connected,'' that is, any two $3$-PVCs in the same dashed box are reconfigurable with each other without changing the vertices in connecting parts.
Furthermore, this gadget preserves the ``external adjacency'' in the following sense: 
if we contract the $3$-PVCs in the same dashed box in Fig.~\ref{figure:and_reconf}(b) into a single vertex, then the resulting graph is exactly the graph depicted in Fig.~\ref{figure:and_reconf}(a). 
Therefore, we can conclude that our \textsc{and} gadget correctly works as an NCL \textsc{and} vertex.

Next we explain the construction of \textsc{or} gadget.
\figurename~\ref{figure:NCLgadget}(b) illustrates our \textsc{or} gadget for the case where $k=3$.
The main part of \textsc{or} gadget forms $C_{k+1}$ (cycle consists of $k+1$ vertices).
Note that we must choose at least two of the vertices on this part to obtain $k$-PVC.
Then we arbitrary choose three distinct vertices from this cycle and connect them to three connecting parts one by one.
\red{
If a weight-2 edge is directed outward, we must choose the vertex in main part next to the connecting part corresponding to the edge to obtain $k$-PVC.
}
To verify that this \textsc{or} gadget correctly simulates an NCL \textsc{or} vertex, it suffices to show that this gadget satisfies both the internal connectedness and the external adjacency.
Since this gadget has only 18 $3$-PVCs where we allow to choose at most five vertices as $3$-PVC.
Therefore, by same way to \textsc{and} gadget, we can easily check these sufficient conditions. \red{(See Fig.~\ref{figure:or_reconf}.)}

\begin{figure}[t]
	\begin{minipage}{0.5\hsize}
		\centering
		\scalebox{0.65}{
			\begin{tikzpicture}
				\node (a1) at (1,2) [circle, draw] {};
				\node (a2) at (2,2) [circle, draw, fill=black] {};
				\node (a3) at (2,1) [circle, draw, fill=gray] {};
				\node (a4) at (1,1) [circle, draw] {};
				\node (b1) at (5,2) [circle, draw] {};
				\node (b2) at (4,2) [circle, draw, fill=black] {};
				\node (b3) at (4,1) [circle, draw, fill=gray] {};
				\node (b4) at (5,1) [circle, draw] {};
				\node (c1) at (4,5) [circle, draw] {};
				\node (c2) at (3,5) [circle, draw, fill=black] {};
				\node (c3) at (3,6) [circle, draw, fill=gray] {};
				\node (c4) at (4,6) [circle, draw] {};
				\node (v1) at (3,4) [circle, draw] {};
				\node (v2) at (3,3) [circle, draw] {};
				\node (v3) at (3,2) [circle, draw] {};
				\draw (a1) -- (a2) -- (a3) -- (a4);
				\draw (b1) -- (b2) -- (b3) -- (b4);
				\draw (c1) -- (c2) -- (c3) -- (c4);
				\draw (v1) -- (v2) -- (v3);
				\draw (a2) -- (v3);
				\draw (b2) -- (v3);
				\draw (c2) -- (v1);
				\draw [densely dashed, thick, rounded corners] ($(a1)+(-0.3,0.3)$) -- ($(a2)+(0.3,0.3)$) -- ($(a3)+(0.3,-0.3)$) -- ($(a4)+(-0.3,-0.3)$) -- cycle;
				\draw [densely dashed, thick, rounded corners] ($(b1)+(0.3,0.3)$) -- ($(b2)+(-0.3,0.3)$) -- ($(b3)+(-0.3,-0.3)$) -- ($(b4)+(0.3,-0.3)$) -- cycle;
				\draw [densely dashed, thick, rounded corners] ($(c1)+(0.3,-0.3)$) -- ($(c2)+(-0.3,-0.3)$) -- ($(c3)+(-0.3,0.3)$) -- ($(c4)+(0.3,0.3)$) -- cycle;
			\end{tikzpicture}
		}
	\end{minipage}
	\begin{minipage}{0.5\hsize}
		\centering
		\scalebox{0.65}{
			\begin{tikzpicture}
				\node (a1) at (1,2) [circle, draw] {};
				\node (a2) at (2,2) [circle, draw, fill=black] {};
				\node (a3) at (2,1) [circle, draw, fill=gray] {};
				\node (a4) at (1,1) [circle, draw] {};
				\node (b1) at (7,2) [circle, draw] {};
				\node (b2) at (6,2) [circle, draw, fill=black] {};
				\node (b3) at (6,1) [circle, draw, fill=gray] {};
				\node (b4) at (7,1) [circle, draw] {};
				\node (c1) at (5,4) [circle, draw] {};
				\node (c2) at (4,4) [circle, draw, fill=black] {};
				\node (c3) at (4,5) [circle, draw, fill=gray] {};
				\node (c4) at (5,5) [circle, draw] {};
				\node (v1) at (3,2) [circle, draw] {};
				\node (v2) at (4,1) [circle, draw] {};
				\node (v3) at (5,2) [circle, draw] {};
				\node (v4) at (4,3) [circle, draw] {};
				\draw (a1) -- (a2) -- (a3) -- (a4);
				\draw (b1) -- (b2) -- (b3) -- (b4);
				\draw (c1) -- (c2) -- (c3) -- (c4);
				\draw (v1) -- (v2) -- (v3) -- (v4) -- (v1);
				\draw (a2) -- (v1);
				\draw (b2) -- (v3);
				\draw (c2) -- (v4);
				\draw [densely dashed, thick, rounded corners] ($(a1)+(-0.3,0.3)$) -- ($(a2)+(0.3,0.3)$) -- ($(a3)+(0.3,-0.3)$) -- ($(a4)+(-0.3,-0.3)$) -- cycle;
				\draw [densely dashed, thick, rounded corners] ($(b1)+(0.3,0.3)$) -- ($(b2)+(-0.3,0.3)$) -- ($(b3)+(-0.3,-0.3)$) -- ($(b4)+(0.3,-0.3)$) -- cycle;
				\draw [densely dashed, thick, rounded corners] ($(c1)+(0.3,-0.3)$) -- ($(c2)+(-0.3,-0.3)$) -- ($(c3)+(-0.3,0.3)$) -- ($(c4)+(0.3,0.3)$) -- cycle;
			\end{tikzpicture}
		}
	\end{minipage}
	\begin{minipage}{0.5\hsize}
		\centering
		(a)
	\end{minipage}
	\begin{minipage}{0.5\hsize}
		\centering
		(b)
	\end{minipage}
	\caption{Gadgets for $3$-PVCR. (a) The \textsc{and} gadget. (b) The \textsc{or} gadget.
        	\red{
               Each dashed rectangle represents a connecting part.
               For each connecting part, choosing the black vertex corresponds to inward direction,  and the gray vertex corresponds to outward direction.
           }
        }
	\label{figure:NCLgadget}
\end{figure}

\begin{figure*}[t]
	\centering
	\scalebox{0.8}{
		\begin{tikzpicture}
			\foreach \x / \y in {-0.2/3.5, 4.3/1.75, 8.8/3.5, 8.8/0, 11.8/0}{
				\node (v0) at (1.0+\x,1.0+\y) [circle, draw] {};
				\node (v1) at (0.2+\x,0.5+\y) [label=above:1] {};
				\node (v2) at (1.8+\x,0.5+\y) [label=above:1] {};
				\node (v3) at (1.0+\x,2.0+\y) [label=below right:2] {};
				\draw (v0) -- (v1) [red] {};
				\draw (v0) -- (v2) [red] {};
				\draw (v0) -- (v3) [ultra thick, blue] {};
			}
			\draw (0.7,4.9) -- (0.8,5.1) -- (0.9,4.9) [ultra thick] {};
			\foreach \x / \y in {4.3/1.75, 8.8/3.5, 8.8/0, 11.8/0}{
				\draw (0.9+\x,1.6+\y) -- (1.0+\x,1.4+\y) -- (1.1+\x,1.6+\y) [ultra thick] {};
			}
			\foreach \x / \y in {-0.2/3.5, 4.3/1.75, 8.8/3.5}{
				\draw (1.53+\x,0.78+\y) -- (1.32+\x,0.8+\y) -- (1.43+\x,0.62+\y) [] {};
			}            
			\foreach \x / \y in {8.8/0, 11.8/0}{
				\draw (1.37+\x,0.88+\y) -- (1.48+\x,0.7+\y) -- (1.27+\x,0.72+\y) [] {};
			}
			\foreach \x / \y in {-0.2/3.5, 4.3/1.75, 8.8/0}{
				\draw (0.47+\x,0.78+\y) -- (0.68+\x,0.8+\y) -- (0.57+\x,0.62+\y) [] {};
			}            
			\foreach \x / \y in {8.8/3.5, 11.8/0}{
				\draw (0.63+\x,0.88+\y) -- (0.52+\x,0.7+\y) -- (0.73+\x,0.72+\y) [] {};
			}            
			\draw (-0.6,2.9) rectangle (2.2,6.1) [thick, dashed, green];
			\draw (2.4,-0.6) rectangle (8.2,6.1) [thick, dashed, green];
			\draw (8.4,-0.6) rectangle (11.2,2.6) [thick, dashed, green];
			\draw (8.4,2.9) rectangle (11.2,6.1) [thick, dashed, green];
			\draw (11.4,-0.6) rectangle (14.2,2.6) [thick, dashed, green];
			\draw (2.2,4.5) -- (2.4,4.5) [ultra thick];
			\draw (8.2,4.5) -- (8.4,4.5) [ultra thick];
			\draw (8.2,1.0) -- (8.4,1.0) [ultra thick];
			\draw (11.2,1.0) -- (11.4,1.0) [ultra thick];
			\draw (11.2,2.9) -- (11.4,2.6) [ultra thick];
		\end{tikzpicture}
	}
	\\(a)\\~\\
	\scalebox{0.8}{
		\begin{tikzpicture}
			\foreach \x / \y in {0/3.5, 3/3.5, 9/3.5, 3/0, 6/0, 9/0, 12/0}{
				\node (a1) at (0.0+\x,0.4+\y) [scale=0.5, circle, draw] {};
				\node (a2) at (0.4+\x,0.4+\y) [scale=0.5, circle, draw] {};
				\node (a3) at (0.4+\x,0.0+\y) [scale=0.5, circle, draw] {};
				\node (a4) at (0.0+\x,0.0+\y) [scale=0.5, circle, draw] {};
				\node (b1) at (1.6+\x,0.4+\y) [scale=0.5, circle, draw] {};
				\node (b2) at (1.2+\x,0.4+\y) [scale=0.5, circle, draw] {};
				\node (b3) at (1.2+\x,0.0+\y) [scale=0.5, circle, draw] {};
				\node (b4) at (1.6+\x,0.0+\y) [scale=0.5, circle, draw] {};
				\node (c1) at (1.2+\x,1.6+\y) [scale=0.5, circle, draw] {};
				\node (c2) at (0.8+\x,1.6+\y) [scale=0.5, circle, draw] {};
				\node (c3) at (0.8+\x,2.0+\y) [scale=0.5, circle, draw] {};
				\node (c4) at (1.2+\x,2.0+\y) [scale=0.5, circle, draw] {};
				\node (v1) at (0.8+\x,1.2+\y) [scale=0.5, circle, draw] {};
				\node (v2) at (0.8+\x,0.8+\y) [scale=0.5, circle, draw] {};
				\node (v3) at (0.8+\x,0.4+\y) [scale=0.5, circle, draw] {};
				\draw (a1) -- (a2) -- (a3) -- (a4);
				\draw (b1) -- (b2) -- (b3) -- (b4);
				\draw (c1) -- (c2) -- (c3) -- (c4);
				\draw (v1) -- (v2) -- (v3);
				\draw (a2) -- (v3);
				\draw (b2) -- (v3);
				\draw (c2) -- (v1);
				\draw (-0.4+\x,-0.4+\y) rectangle (2.0+\x,2.4+\y); 
			}
			\draw (2.0,4.5) -- (2.6,4.5);
			\draw (5.0,1.0) -- (5.6,1.0);
			\draw (8.0,1.0) -- (8.6,1.0);
			\draw (11.0,1.0) -- (11.6,1.0);
			\draw (3.8,3.1) -- (3.8,2.4);
			\draw (8.6,3.1) -- (8.0,2.4);
			\draw (11.6,2.4) -- (11.0,3.1);
			\draw (5.0,3.1) -- (5.6,2.4) [dashed];
			\foreach \x / \y in {0.4/3.9, 0.8/5.5, 0.8/4.7, 1.2/3.9, 3.4/3.9, 3.8/5.1, 3.8/4.7, 4.2/3.9, 3.4/0.4, 3.8/1.6, 3.8/0.8, 4.2/0.4, 6.4/0.4, 6.8/1.6, 6.8/0.4, 7.2/0.4, 9.4/3.5, 9.8/5.1, 9.8/3.9, 10.2/3.9, 9.4/0.4, 9.8/1.6, 9.8/0.4, 10.2/0.0, 12.4/0.0, 12.8/1.6, 12.8/0.4, 13.2/0.0}{
				\node (t) at (\x,\y) [scale=0.5, circle, draw, fill=black] {};
			}
			\draw (-0.6,2.9) rectangle (2.2,6.1) [thick, dashed, green];
			\draw (2.4,-0.6) rectangle (8.2,6.1) [thick, dashed, green];
			\draw (8.4,-0.6) rectangle (11.2,2.6) [thick, dashed, green];
			\draw (8.4,2.9) rectangle (11.2,6.1) [thick, dashed, green];
			\draw (11.4,-0.6) rectangle (14.2,2.6) [thick, dashed, green];
		\end{tikzpicture}
	}
	\\(b)
	\caption{(a) All vaild orientations of the edges incident to an NCL \textsc{and} vertex, and (b) all $3$-PVCs of the \textsc{and} gadget. The $3$-PVCs connected by an edge are adjacent by $\mathsf{TJ}$/$\mathsf{TS}$ rules, while the $3$-PVCs connected by dashed edge are adjacent only by $\mathsf{TJ}$ rule.}
	\label{figure:and_reconf}
\end{figure*}

\begin{figure*}[t]
		\centering
		\scalebox{0.65}{
			\begin{tikzpicture}
				\coordinate (coa_1) at (0,0);
				\coordinate (coa_2) at (3* 0.8660,2* 0.5);
				\coordinate (coa_3) at (3* 0.0000,2* 1.0);
				\coordinate (coa_4) at (3*-0.8660,2* 0.5);
				\coordinate (coa_5) at (3*-0.8660,2*-0.5);
				\coordinate (coa_6) at (3* 0.0000,2*-1.0);
				\coordinate (coa_7) at (3* 0.8660,2*-0.5);
				\draw ($(coa_1)+(1.0,1.0)$) -- ($(coa_3)+(1.0,1.0)$) [ultra thick] {};
				\draw ($(coa_1)+(1.0,1.0)$) -- ($(coa_5)+(1.0,1.0)$) [ultra thick] {};
				\draw ($(coa_1)+(1.0,1.0)$) -- ($(coa_7)+(1.0,1.0)$) [ultra thick] {};
				\draw ($(coa_2)+(1.0,1.0)$) -- ($(coa_3)+(1.0,1.0)$) -- ($(coa_4)+(1.0,1.0)$) -- ($(coa_5)+(1.0,1.0)$) -- ($(coa_6)+(1.0,1.0)$) -- ($(coa_7)+(1.0,1.0)$) -- ($(coa_2)+(1.0,1.0)$) [ultra thick] {};
				\foreach \n in {1, 2, 3, 4, 5, 6, 7} {
					\draw ($(coa_\n)+(0,0.3)$) rectangle ($(coa_\n)+(2.0,2.0)$) [thick, dashed, green, fill=white];
					\node (\n_v0) at ($(coa_\n)+(1.0,1.0)$) [circle, draw] {};
					\node (\n_v1) at ($(coa_\n)+(0.2,0.5)$) [label=above:2] {};
					\node (\n_v2) at ($(coa_\n)+(1.8,0.5)$) [label=above:2] {};
					\node (\n_v3) at ($(coa_\n)+(1.0,2.0)$) [label=below right:2] {};
					\draw (\n_v0) -- (\n_v1) [ultra thick, blue] {};
					\draw (\n_v0) -- (\n_v2) [ultra thick, blue] {};
					\draw (\n_v0) -- (\n_v3) [ultra thick, blue] {};
				}
				\foreach \n in {2, 3, 4} \draw ($(coa_\n)+(0.9,1.4)$) -- ($(coa_\n)+(1.0,1.6)$) -- ($(coa_\n)+(1.1,1.4)$) [ultra thick] {};
				\foreach \n in {1, 5, 6, 7} \draw ($(coa_\n)+(0.9,1.6)$) -- ($(coa_\n)+(1.0,1.4)$) -- ($(coa_\n)+(1.1,1.6)$) [ultra thick] {};
				\foreach \n in {2, 6, 7} \draw ($(coa_\n)+(1.37,0.88)$) -- ($(coa_\n)+(1.48,0.7)$) -- ($(coa_\n)+(1.27,0.72)$) [ultra thick] {};
				\foreach \n in {1, 3, 4, 5} \draw ($(coa_\n)+(1.53,0.78)$) -- ($(coa_\n)+(1.32,0.8)$) -- ($(coa_\n)+(1.43,0.62)$) [ultra thick] {};
				\foreach \n in {4, 5, 6} \draw ($(coa_\n)+(0.63,0.88)$) -- ($(coa_\n)+(0.52,0.7)$) -- ($(coa_\n)+(0.73,0.72)$) [ultra thick] {};
				\foreach \n in {1, 2, 3, 7} \draw ($(coa_\n)+(0.47,0.78)$) -- ($(coa_\n)+(0.68,0.8)$) -- ($(coa_\n)+(0.57,0.62)$) [ultra thick] {};
			\end{tikzpicture}
		}
		\\ (a)
		\\ ~
		\\
		\scalebox{0.65}{
			\begin{tikzpicture}[scale = 0.3]
				\coordinate (a_11) at (12* 0.8660,12* 0.5);
				\coordinate (a_12) at (12* 0.0000,12* 1.0);
				\coordinate (a_13) at (12*-0.8660,12* 0.5);
				\coordinate (a_14) at (12*-0.8660,12*-0.5);
				\coordinate (a_15) at (12* 0.0000,12*-1.0);
				\coordinate (a_16) at (12* 0.8660,12*-0.5);
				
				\coordinate (a_21) at (24* 0.8660,24* 0.5000);
				\coordinate (a_22) at (24* 0.5000,24* 0.8660);
				\coordinate (a_23) at (24* 0.0000,24* 1.0000);
				\coordinate (a_24) at (24*-0.5000,24* 0.8660);
				\coordinate (a_25) at (24*-0.8660,24* 0.5000);
				\coordinate (a_26) at (24*-1.0000,24* 0.0000);
				\coordinate (a_27) at (24*-0.8660,24*-0.5000);
				\coordinate (a_28) at (24*-0.5000,24*-0.8660);
				\coordinate (a_29) at (24* 0.0000,24*-1.0000);
				\coordinate (a_30) at (24* 0.5000,24*-0.8660);
				\coordinate (a_31) at (24* 0.8660,24*-0.5000);
				\coordinate (a_32) at (24* 1.0000,24* 0.0000);
				
				\coordinate (a1) at (1,2);
				\coordinate (a2) at (2,2);
				\coordinate (a3) at (2,1);
				\coordinate (a4) at (1,1);
				\coordinate (b1) at (7,2);
				\coordinate (b2) at (6,2);
				\coordinate (b3) at (6,1);
				\coordinate (b4) at (7,1);
				\coordinate (c1) at (5,4);
				\coordinate (c2) at (4,4);
				\coordinate (c3) at (4,5);
				\coordinate (c4) at (5,5);
				\coordinate (v1) at (3,2);
				\coordinate (v2) at (4,1);
				\coordinate (v3) at (5,2);
				\coordinate (v4) at (4,3);

				\foreach \n/\m in {11/12, 11/15, 11/22, 11/32, 12/13, 12/14, 12/16, 12/23, 13/15, 13/24, 13/26, 14/15, 14/27, 15/16, 15/28, 15/30, 16/31, 21/22, 21/32, 22/23, 23/24, 24/25, 25/26, 27/28, 28/29, 29/30, 30/31} \draw ($(a_\n)+(4,3)$) -- ($(a_\m)+(4,3)$);
				\foreach \n/\m in {11/13, 11/16, 13/14, 14/16, 26/27, 31/32} \draw ($(a_\n)+(4,3)$) -- ($(a_\m)+(4,3)$) [dashed];
				\draw (24*-0.5000+4,24* 0.9330) -- (24* 0.5000+4,24* 0.9330) [dashed];
				\draw (24* 0.5000+4,24*-0.8660+3) -- (24* 0.9330+8,24*-0.8660+3) -- (24* 0.9330+8,24* 0.0000+3);
				\draw (24*-0.9330,24* 0.0000+3) -- (24*-0.9330,24*-0.8660+3) -- (24*-0.5000+4,24*-0.8660+3);

				\foreach \n in {11,12,13,14,15,16,21,22,23,24,25,26,27,28,29,30,31,32}{
					\draw ($(a_\n)+(0,0)$) rectangle ($(a_\n)+(8,6)$) [fill=white];
					\node (\n_a1) at ($(a_\n)+(a1)$) [circle, draw, scale=0.5] {};
					\node (\n_a4) at ($(a_\n)+(a4)$) [circle, draw, scale=0.5] {};
					\node (\n_b1) at ($(a_\n)+(b1)$) [circle, draw, scale=0.5] {};
					\node (\n_b4) at ($(a_\n)+(b4)$) [circle, draw, scale=0.5] {};
					\node (\n_c1) at ($(a_\n)+(c1)$) [circle, draw, scale=0.5] {};
					\node (\n_c4) at ($(a_\n)+(c4)$) [circle, draw, scale=0.5] {};
				}
				\foreach \n in {11,12,13,14,15,16,21,22,23,24,30,31,32}{
					\node (\n_a2) at ($(a_\n)+(a2)$) [circle, draw, scale=0.5, fill=black] {};
					\node (\n_a3) at ($(a_\n)+(a3)$) [circle, draw, scale=0.5] {};
				}
				\foreach \n in {25,26,27,28,29}{
					\node (\n_a2) at ($(a_\n)+(a2)$) [circle, draw, scale=0.5] {};
					\node (\n_a3) at ($(a_\n)+(a3)$) [circle, draw, scale=0.5, fill=black] {};
				}
				\foreach \n in {11,12,13,14,15,16,22,23,24,25,26,27,28}{
					\node (\n_b2) at ($(a_\n)+(b2)$) [circle, draw, scale=0.5, fill=black] {};
					\node (\n_b3) at ($(a_\n)+(b3)$) [circle, draw, scale=0.5] {};
				}
				\foreach \n in {21,29,30,31,32}{
					\node (\n_b2) at ($(a_\n)+(b2)$) [circle, draw, scale=0.5] {};
					\node (\n_b3) at ($(a_\n)+(b3)$) [circle, draw, scale=0.5, fill=black] {};
				}
				\foreach \n in {11,12,13,14,15,16,26,27,28,29,30,31,32}{
					\node (\n_c2) at ($(a_\n)+(c2)$) [circle, draw, scale=0.5, fill=black] {};
					\node (\n_c3) at ($(a_\n)+(c3)$) [circle, draw, scale=0.5] {};
				}
				\foreach \n in {21,22,23,24,25}{
					\node (\n_c2) at ($(a_\n)+(c2)$) [circle, draw, scale=0.5] {};
					\node (\n_c3) at ($(a_\n)+(c3)$) [circle, draw, scale=0.5, fill=black] {};
				}
				\foreach \n in {13,14,15,24,25,26,27,28,29,30} \node (\n_v1) at ($(a_\n)+(v1)$) [circle, draw, scale=0.5, fill=black] {};
				\foreach \n in {11,12,16,21,22,23,31,32} \node (\n_v1) at ($(a_\n)+(v1)$) [circle, draw, scale=0.5] {};
				\foreach \n in {12,14,16,23,27,31} \node (\n_v2) at ($(a_\n)+(v2)$) [circle, draw, scale=0.5, fill=black] {};
				\foreach \n in {11,13,15,21,22,24,25,26,28,29,30,32} \node (\n_v2) at ($(a_\n)+(v2)$) [circle, draw, scale=0.5] {};
				\foreach \n in {11,15,16,21,22,28,29,30,31,32} \node (\n_v3) at ($(a_\n)+(v3)$) [circle, draw, scale=0.5, fill=black] {};
				\foreach \n in {12,13,14,23,24,25,26,27} \node (\n_v3) at ($(a_\n)+(v3)$) [circle, draw, scale=0.5] {};
				\foreach \n in {11,12,13,21,22,23,24,25,26,32} \node (\n_v4) at ($(a_\n)+(v4)$) [circle, draw, scale=0.5, fill=black] {};
				\foreach \n in {14,15,16,27,28,29,30,31} \node (\n_v4) at ($(a_\n)+(v4)$) [circle, draw, scale=0.5] {};
				\foreach \n in {11,12,13,14,15,16,21,22,23,24,25,26,27,28,29,30,31,32}{
					\draw (\n_a1) -- (\n_a2) -- (\n_a3) -- (\n_a4);
					\draw (\n_b1) -- (\n_b2) -- (\n_b3) -- (\n_b4);
					\draw (\n_c1) -- (\n_c2) -- (\n_c3) -- (\n_c4);
					\draw (\n_v1) -- (\n_v2) -- (\n_v3) -- (\n_v4) -- (\n_v1);
					\draw (\n_a2) -- (\n_v1);
					\draw (\n_b2) -- (\n_v3);
					\draw (\n_c2) -- (\n_v4);
				}
				\draw ($(a_21)+(-0.5,-0.5)$) rectangle ($(a_21)+(8.5,6.5)$) [thick, dashed, green];
				\draw ($(a_25)+(-0.5,-0.5)$) rectangle ($(a_25)+(8.5,6.5)$) [thick, dashed, green];
				\draw ($(a_29)+(-0.5,-0.5)$) rectangle ($(a_29)+(8.5,6.5)$) [thick, dashed, green];
				\draw (12*-0.8660-0.5,12*-1.0-0.5) rectangle (12* 0.8660+8.5,12* 1.0+6.5) [thick, dashed, green];
				\draw (24*-0.5000-0.5,24* 0.8660-0.5) rectangle (24* 0.5000+8.5,24* 1.0000+6.5) [thick, dashed, green];
				\draw (24*-1.0000-0.5,24* 0.0000+6.5) -- (24*-1.0000-0.5,24*-0.8660-0.5) -- (24*-0.5000+8.5,24*-0.8660-0.5) -- (24*-0.5000+8.5,24*-0.8660+6.5) -- (24*-0.8660+8.5,24*-0.8660+6.5) -- (24*-0.8660+8.5,24* 0.0000+6.5) -- (24*-1.0000-0.5,24* 0.0000+6.5) [thick, dashed, green];
				\draw (24* 1.0000+8.5,24* 0.0000+6.5) -- (24* 1.0000+8.5,24*-0.8660-0.5) -- (24* 0.5000-0.5,24*-0.8660-0.5) -- (24* 0.5000-0.5,24*-0.8660+6.5) -- (24* 0.8660-0.5,24*-0.8660+6.5) -- (24* 0.8660-0.5,24* 0.0000+6.5) -- (24* 1.0000+8.5,24* 0.0000+6.5) [thick, dashed, green];
			\end{tikzpicture}
		}
		\\ (b)
	\caption{
		\red{(a) All vaild orientations of the edges incident to an NCL \textsc{or} vertex, and (b) all $3$-PVCs of the \textsc{or} gadget.
		The $3$-PVCs connected by an edge are adjacent by $\mathsf{TJ}$/$\mathsf{TS}$ rules, while the $3$-PVCs connected by dashed edge are adjacent only by $\mathsf{TJ}$ rule.}}
	\label{figure:or_reconf}
\end{figure*}
	
\subsubsection{Reduction} 

As we have explained before, we replace each of NCL \textsc{and}/\textsc{or} vertices with its corresponding gadget; let $G$ be the resulting graph.
Recall that NCL remains $\mathtt{PSPACE}$-complete even if an input NCL machine is planar and bounded bandwidth~\cite{Z15}.
Since both our gadgets are planar, consist of only a constant number of edges, and of maximum degree three, the resulting graph $G$ is also planar, bounded bandwidth and of maximum degree three.
(In fact the number of edges in our gadget is $O(k)$. However, since we defined that $k$ is a fixed integer, it becomes constant.)

In addition, we construct two $k$-PVCs of $G$ which correspond to two given NCL configurations of the NCL machine.
Note that there are (in general, exponentially) many $k$-PVCs which correspond to the same NCL configuration.
However, by the construction of the gadgets, no two distinct NCL configurations correspond to the same $k$-PVC of $G$.
Therefore, we arbitrarily choose two $k$-PVCs of $G$ which correspond to two given NCL configurations.

This completes the construction of our corresponding instance of $k$-PVCR.
Clearly the construction can be done in polynomial time.

\subsubsection{Correctness} 

Let $C_I$ and $C_J$ be two given NCL configurations of the NCL machine.
Let $I$ and $J$ be two $k$-PVCs of $G$ which correspond to $C_I$ and $C_J$, respectively.
We now prove that there exists a desired sequence of NCL configurations between $C_I$ and $C_J$ if and only if there exists a reconfiguration sequence between $I$ and $J$.

We first prove the only-if direction.
Suppose that there exists a desired sequence $S = \langle C_0, C_1, \dots, C_\ell \rangle$ of NCL configurations between $C_0 = C_I$ and $C_\ell = C_J$.
Consider any two adjacent NCL configurations $C_{i-1}$ and $C_i$ in the sequence.
Then only one NCL edge $vw$ changes its orientation between $C_{i-1}$ and $C_i$.
Notice that, since both $C_{i-1}$ and $C_i$ are valid NCL configurations, the NCL \textsc{and}/\textsc{or} vertices $v$ and $w$ have enough in-coming NCL edges even without $vw$.
Recall that both \textsc{and}/\textsc{or} gadgets are internally connected and preserve the external adjacency.
Therefore, any reversal of an NCL edge can be simulated by a reconfiguration sequence of $k$-PVCs of $G$, and hence there exists a reconfiguration sequence between $I$ and $J$.

We now prove the if direction.
Suppose that there exists a reconfiguration sequence $S = \langle I_0, I_1, \dots, I_\ell \rangle$ ($I_0 = I$ and $I_\ell = J$).
Notice that, by the construction of gadgets, any $k$-PVC of $G$ corresponds to a valid NCL configuration.
Let $C_i$ be an NCL configuration corresponds to $I_i$, for $i \in \{ 0, \dots , \ell \} $.
By deleting redundant orientations from $C_0, C_1, \dots , C_\ell$ if needed, we can obtain a sequence of valid adjacent orientations between $C_I$ and $C_J$.

This completes the proof of Theorem~\ref{thm:PSPACEC-planar-maxdeg3}.

\section{Polynomial-Time Algorithms}
\label{sec:polytime-algos}

\subsection{Trees}
\label{sec:trees}

In this section, we show polynomial-time algorithms for \textsc{$k$-PVCR} on trees under each of $\mathsf{TJ}$ and $\mathsf{TAR}$. We first show a polynomial-time algorithm for the problem under $\mathsf{TJ}$. Then, using Lemma~\ref{lem:TJ-TAR-equiv} and the above result, we show a polynomial-time algorithm for the problem under $\mathsf{TAR}$.

First, in order to solve the problem under $\mathsf{TJ}$, we claim that for an instance $(T, I, J, \mathsf{TJ})$ of \textsc{$k$-PVCR} on a tree $T$, if $|I| = |J|$, one can construct in polynomial time a $\mathsf{TJ}$-sequence between $I$ and $J$. The idea is to construct a canonical $k$-path vertex cover $I^\star$ such that both $I$ and $J$ can be reconfigured to $I^\star$ under $\mathsf{TJ}$. 

Before constructing $I^\star$, we prove the following lemma, which describes an useful algorithm for partitioning a tree into subtrees satisfying certain conditions.
\begin{lemma}\label{lem:tree-partition}
	Let $T$ be a tree on $n$ vertices rooted at a vertex $r$.
	Assume that $\psi_k(T) \geq 1$.
	Then, in $O(n)$ time, one can partition $T$ into $\psi_k(T)$ subtrees $T_1(r), \dots,\allowbreak T_{\psi_k(T)}(r)$ such that for each $i \in \{1, \dots, \psi_k(T)\}$,
	\begin{itemize}
		\item[(i)] Each $k$-path vertex cover $I$ satisfies $I \cap V(T_i(r)) \neq \emptyset$.
		\item[(ii)] There is a vertex that covers all $k$-paths in $T_i(r)$.
	\end{itemize}
\end{lemma}

\begin{proof}
	To construct a partition $P(T) = \{T_1(r), \dots, T_{\psi_k(T)}(r)\}$ of $T$ satisfying the described conditions, we slightly modify the algorithm $\mathtt{PVCPTree}(T, k)$ in~\cite{BrevsarKKS11} as follows.
	A \emph{properly rooted subtree} $T_v$ of $T$ is a subtree of $T$ induced by the vertex $v$ and all its descendants (with respect to the root $r$) satisfying the following conditions
	\begin{enumerate}
		\item $T_v$ contains a $k$-path;
		\item $T_v - v$ does not contain a $k$-path.
	\end{enumerate} 
	The modified algorithm $\mathtt{Partition}(T, k, r)$ systematically searches for a properly rooted tree $T_v$, decides whether $T_v$ belongs to a solution $P(T)$, and if so, \red{adds} $T_v$ to $P(T)$, and \red{removes} $T_v$ from the input tree $T$. 
	\red{To check if $T$ contains a properly rooted subtree $T_v$, one can start by assigning $v$ to a vertex of largest \textit{depth} (i.e., distance from $r$) and verify if $T_v$ is properly rooted. If so, we return ``yes''. Otherwise, we assign $v$ to its parent and repeat, until a $T_v$ is found (returning ``yes'') or there is nothing to check (returning ``no'').}
	
	\begin{algorithm}[t]
		\KwIn{A tree $T$ on $n$ vertices rooted at $r$ and a positive integer $k$\;}
		\KwOut{A partition $P(T)$ of $T$ into $\psi_k$(T) subtrees\;}
		\SetArgSty{textbb}   
		$i := 1$\;
		\While{$T$ contains a properly rooted subtree $T_v$}
		{            
			\If{$T - T_v$ contains a properly rooted subtree}
			{
				$T_i(r) := T_v$\;
				$i := i + 1$\;
			}
			\Else{
				$T_i(r) := T$\;
			}
			$T := T - T_v$\;
		}
		$P(T) = \{T_1(r), \dots, T_i(r)\}$\;
		\Return $P(T)$\;
		\caption{$\mathtt{Partition}(T, k, r)$.}
	\end{algorithm}
	
	From~\cite{BrevsarKKS11}, it follows that $\mathtt{Partition}(T, k, r)$ runs in $O(n)$ time.
	From the construction of $P(T)$, it is clear that (i) always holds. 
	We show (ii) by induction on $\psi_k(T)$.
	
	For a tree $T$ with $\psi_k(T) = 1$, let $T_v$ be a properly rooted subtree of $T$.
	Since any $k$-path vertex cover of $T$ contains a vertex from $T_v$, it follows that $\psi_k(T - T_v) = \psi_k(T) - 1 = 0$, which implies that $T - T_v$ does not contain any properly rooted subtree, and therefore $P(T) = \{T\}$.
	To see that (ii) holds, note that $v$ must cover all $k$-paths in $T_v$, and therefore it also covers all $k$-paths in $T$; otherwise, $T - T_v$ contains a $k$-path that is not covered by $v$, and then must contain a properly rooted subtree, which is a contradiction.
	
	Assume that (ii) holds for any tree $T$ with $\psi_k(T) < c$, for some constant $c > 1$.
	For a tree $T$ rooted at some vertex $r$ with $\psi_k(T) = c$, let $T_v$ be a properly rooted subtree of $T$, where $v$ is some vertex of $T$.
	From the algorithm $\mathtt{Partition}$, it follows that $v$ must cover all $k$-paths in $T_v = T_1(r)$.
	Since $c > 1$, the tree $T - T_v$ contains a properly rooted subtree.
	By inductive hypothesis, for each $i \in \{2, 3, \dots, \psi_k(T)\}$, there is a vertex that covers all $k$-paths in $T_i(r)$.
	Therefore, (ii) holds for any tree $T$ with $\psi_k(T) \geq 1$.
	\end{proof}

We are now ready to show the following theorem.
\begin{theorem}
	\label{thm:kPVC-Trees-TJ}
	For any instance $(T, I, J, \mathsf{TJ})$ of \textsc{$k$-PVCR} on a tree $T$, $I$ and $J$ are reconfigurable if and only if $|I| = |J|$.
	Moreover, a reconfiguration sequence between them, if exists, can be constructed in $O(n)$ time.
	Consequently, \textsc{$k$-PVCR} under $\mathsf{TJ}$ can be solved in linear time on trees.
\end{theorem}

\begin{proof}
	Clearly, if $I$ and $J$ are reconfigurable under $\mathsf{TJ}$, they must be of the same size.
	To prove this theorem, it suffices to show that for an instance $(T, I, J, \mathsf{TJ})$ of \textsc{$k$-PVCR} on a tree $T$, one can construct in polynomial time a $\mathsf{TJ}$-sequence between $I$ and $J$.
	
	A minimum $k$-path vertex cover $I^r$ can be easily constructed in linear time by modifying $\mathtt{Partition}$ as follows: Initially, $I^r = \emptyset$. In each iteration of the {\bf while} loop, add to $I^r$ the vertex $v$ of the properly rooted subtree $T_v$ that is currently considering.
	\red{Such a vertex $v$ can be obtained from the process of checking if $T$ contains a properly rooted subtree described in the proof of Lemma~\ref{lem:tree-partition}.} 
	Let $I^\star$ be any $k$-path vertex cover of size $|I| = |J|$ such that $I^r \subseteq I^\star$.
	We claim that both $I$ and $J$ can be reconfigured to $I^\star$ under $\mathsf{TJ}$. 
	As a result, a $\mathsf{TJ}$-sequence between $I$ and $J$ can be constructed by reconfiguring $I$ to $I^\star$, and then $I^\star$ to $J$.
	
	We now show how to construct a $\mathsf{TJ}$-sequence between $I$ and $I^\star$.
	Let $P(T) = \{T_1(r), \dots, \allowbreak T_{\psi_k(T)}(r)\}$ be a partition of $T$ resulting from the algorithm $\mathtt{Partition}$ and let $I_0 = I$. 
	\red{Intuitively, we will first ``settle'' the tokens in $I^r \subseteq I^\star$ (\textbf{Step~1}), and then, as the tokens in $I^r$ already cover all $k$-paths in $T$, the remaining tokens in $I^\star \setminus I^r$ can be easily ``settled'' by jumping tokens one-by-one in arbitrary order (\textbf{Step~2}).} 
	\begin{itemize}
		\item {\bf Step~1:} For \red{each $i$ from $1$ to $\psi_k(T)$}, let $v_i \in I^r \cap V(T_i(r))$.
		If $v_i$ does not contain a token in $I_{i-1}$, we jump a token from some vertex $x_i \in I_{i-1} \cap V(T_i(r))$ to $v_i$.
		Otherwise, we do nothing.
		Let $I_i$ be the resulting set.
		Note that any $k$-path in $T$ covered by $x_i$ must also be covered by some $v_j$ with $j \leq i$.
		A simple induction shows that $I_i = I_{i-1} \setminus \{x_i\} \cup \{v_i\}$ forms a $k$-path vertex cover of $T$.
		\item {\bf Step~2:}
		For $x \in I_{\psi_k(T)} \setminus I^\star$ and $y \in I^\star \setminus I_{\psi_k(T)}$, we simply jump the token on $x$ to $y$, and repeat the process with $I_{\psi_k(T)} \setminus \{x\}$ and $I^\star \setminus \{y\}$ instead of $I_{\psi_k(T)}$ and $I^\star$, respectively.
		Since $I^r \subseteq I_{\psi_k(T)} \cap I^\star$ is already a minimum $k$-path vertex cover, any $\mathsf{TJ}$-step described above results a $k$-path vertex cover of $T$.
	\end{itemize}
	Since each token in $I$ is jumped at most once, the above construction can be done in linear time.
	We have described how to construct a $\mathsf{TJ}$-sequence from $J$ to $I^\star$.
	In a similar manner, a $\mathsf{TJ}$-sequence between $J$ and $I^\star$ can be constructed.
	Our proof of Theorem~\ref{thm:kPVC-Trees-TJ} is complete.
	\end{proof}

Consequently, combining Theorem~\ref{thm:kPVC-Trees-TJ} and Lemma~\ref{lem:TJ-TAR-equiv}, we have the following theorem.
\begin{theorem}
	\label{thm:kPVCR-Trees-TAR}
	For any instance $(T, I, J, \mathsf{TAR}(u))$ of \textsc{$k$-PVCR} on a tree $T$, one can decide if $I$ and $J$ are reconfigurable in polynomial time.
\end{theorem}
\begin{proof}
	Clearly, if $u < \max\{|I|, |J|\}$ or $u = \psi_k(T)$ then $(T, I, J, \mathsf{TAR}(u))$ is a no-instance, because either $I$ or $J$ cannot be modified by adding/removing tokens.
	We now consider the case $u \geq \max\{|I|, |J|\}$ and $u > \psi_k(T)$.
	Note that if $|I| < |J|$ then we can add tokens to $I$ until the resulting $k$-path vertex cover is of size $|J|$, simply because $u \geq \max\{|I|, |J|\}$.
	As a result, we can assume without loss of generality that $|I| = |J| = s$ for some constant $s$.
	By Theorem~\ref{thm:kPVC-Trees-TJ} and Lemma~\ref{lem:TJ-TAR-equiv}, it follows that there always exists a $\mathsf{TAR}(s+1)$-sequence between $I$ and $J$.
	If $s+1 \leq u$ then clearly a $\mathsf{TAR}(s+1)$-sequence is also a $\mathsf{TAR}(u)$-sequence, and we are done.
	Assume that $s+1 > u$.
	Since $u \geq s$ and $u > \psi_k(T)$, it follows that $u = s$ and both $I$ and $J$ are not minimum.
	Now, we need to check if we can remove at least one token from $I$ (resp. $J$), which can be done in polynomial time by checking each token one by one and verifying whether its removal results a $k$-path vertex cover.
	If this is possible for both $I$ and $J$, we remove exactly one token from $I$ (resp. $J$) to obtain a new $k$-path vertex cover $I^\prime$ (resp. $J^\prime$) of size $s-1$.
	By Lemma~\ref{lem:TJ-TAR-equiv}, there exists a $\mathsf{TAR}(u)$-sequence between $I^\prime$ and $J^\prime$, and combining this sequence with the previous removal steps gives us a $\mathsf{TAR}(u)$-sequence between $I$ and $J$.
	Otherwise, we can conclude that the given instance is a no-instance, because the first step of reconfiguring (either from $I$ to $J$ or vice versa) is to remove some token (since $u = s$, adding a token is not possible).
\end{proof}

\subsection{Paths and Cycles}
\label{sec:paths-cycles}
Here, we describe polynomial-time algorithms for \textsc{$k$-PVCR} on paths and cycles.
As paths and cycles are the only (planar) graphs of maximum degree \red{two}, by combining Theorem~\ref{thm:PSPACEC-planar-maxdeg3} and our results, we have a complexity dichotomy of \textsc{$k$-PVCR} on (planar) graphs.
Additionally, on paths, we claim that one can construct a \emph{shortest} reconfiguration sequence between any two given $k$-path vertex covers (if exists) under each reconfiguration rule $\mathsf{TS}$, $\mathsf{TJ}$, and $\mathsf{TAR}$.

\subsubsection{\textsc{$k$-PVCR} on Paths}
By Theorems~\ref{thm:kPVC-Trees-TJ} and \ref{thm:kPVCR-Trees-TAR}, clearly \textsc{$k$-PVCR} on paths can be solved in polynomial time under each of $\mathsf{TJ}$ and $\mathsf{TAR}$. In this section, we slightly improve this result by showing that one can construct a \emph{shortest} reconfiguration sequence between two $k$-path vertex covers on a path not only under each of $\mathsf{TJ}$ and $\mathsf{TAR}$ but also under $\mathsf{TS}$. 

Given an instance $(P,I,J,\mathsf{TJ})$ of \textsc{$k$-PVCR} where $|I|=|J|=s$, one can construct a shortest $\mathsf{TJ}$-sequence between $I$ and $J$. 
Suppose that vertices in $I = \{v_{i_1}, \dots, \allowbreak v_{i_s}\}$ and $J = \{v_{j_1}, \dots, v_{j_s}\}$ are ordered such that $1 \leq i_1 < \dots < i_s \leq n$ and  $1 \leq j_1 < \dots < j_s \leq n$.
In each step of the algorithm, we move a token on the ``rightmost'' vertex $v_{i_p} \in I\setminus J$ to the ``rightmost'' vertex $ v_{j_p} \in J\setminus I$ if $i_p > j_p$ or vice-versa otherwise, for $p \in \{1, \dots, s\}$. 
As a reconfiguration sequence is reversible, one can easily form a $\mathsf{TJ}$-sequence between $I$ and $J$. Note that each step of the algorithm reduces $|I\Delta J|/2$ by exactly one. Finally, we obtain a shortest $\mathsf{TJ}$-sequence between $I$ and $J$ of length exactly $|I\Delta J|/2$. 

\begin{theorem}
	\label{thm:kPVCR_paths_TJ} 
	Given an instance $(P, I, J, \mathsf{TJ})$ of \textsc{$k$-PVCR} on a path $P$, the $k$-path vertex covers $I$ and $J$ are reconfigurable if and only if $|I| = |J|$.
	Moreover, we can compute a shortest reconfiguration sequence in $O(n)$ time. 
\end{theorem}
\begin{proof}
	
	Let $P= v_1v_2\dots v_n$ be a given path. 
	In the following, we use the expression \emph{rightmost} instead of using ``with the largest index''. 
	Algorithm~\ref{algo:PVCRPathTJ} describes an algorithm $\mathtt{PVCRPathTJ}(P, I, J)$ for \textsc{$k$-PVCR} on paths under $\mathsf{TJ}$. 
	
	\begin{algorithm}[ht]
		\KwIn{A path $P$ of $n$ vertices, initial token-set $I$, and target token-set $J$\;}
		\KwOut{A reconfiguration sequence $S$\;}
		\SetArgSty{textbb}
		Let $S, S_I, S_J$ be reconfiguration sequences, and initialize them by $\emptyset$\;
		\While{$I\Delta J\neq \emptyset$}{
			$v_i \leftarrow$ the rightmost vertex in $P[I\Delta J]$\;
			
			\If{$v_i\in I$}
			{
				Find the rightmost token $v_j$ in $J\setminus I$ (here $j<i$)\; 
				$S_I:=S_I \oplus \langle I, I \setminus \{v_i\} \cup \{v_j\} \rangle$\;
				$I:= I \setminus \{v_i\} \cup \{v_j\}$\;
			}
			\If{$v_i\in J$}
			{
				Find the rightmost token $v_j$ in $I\setminus J$ (here $j<i$)\; 
				$S_J:=S_J \oplus \langle J, J \setminus \{v_i\} \cup \{v_j\} \rangle$\;
				$J:=J \setminus \{v_i\} \cup \{v_j\}$\;
			}
		}
		$S:=S_I\oplus \text{rev}(S_J)$\;
		\Return $S$\;
		\caption{$\mathtt{PVCRPathTJ}(P,I,J)$}
		\label{algo:PVCRPathTJ}
	\end{algorithm}
	
	Clearly, if $I$ and $J$ are reconfigurable under $\mathsf{TJ}$ then they are of the same size.
	It remains to show the if direction.
	To this end, we show that $\mathtt{PVCRPathTJ}(P,I,J)$ correctly constructs a $\mathsf{TJ}$-sequence between two $k$-path vertex covers $I, J$ of the same size.
	In each iteration of the {\bf while} loop, when $v_i \in I$, we confirm that if we move a token from $v_i$ to $v_j$, the resulting token-set still keeps $k$-path vertex cover property. 
	In other words, the constructed sequence $S_I$ is indeed a $\mathsf{TJ}$-sequence.
	Suppose to the contrary that moving the token on $v_i$ to the left (i.e., to the direction in which the indices get smaller) results in some non-covered $k$-path, say $Q = v_{\ell} v_{\ell+1}\dots v_{\ell+k-1}$, where $\ell \leq i\leq \ell+k-1$ and $j+1\leq \ell \leq n-k+1$.
	Since $J$ is a $k$-path vertex cover, there must be some vertex $v_{\ell'}\in J$ for $\ell \leq \ell' \leq \ell +k-1$. Also, since $v_i\in I\setminus J$, $\ell'\neq i$. If $\ell'<i$, then $v_{\ell'}\in I$; otherwise, $v_j\in J\setminus I$ is not rightmost. If $\ell'>i$, then $v_{\ell'}\in I$; otherwise, $v_i$ is not rightmost in $P[I\Delta J]$. Therefore $v_{\ell'}\in J\cap I$ always covers $P$, a contradiction. 
	In a similar manner, one can also verify that $S_J$ is indeed a $\mathsf{TJ}$-sequence.
	Let $I^\prime$ be the $k$-path vertex cover obtained when the condition of the {\bf while} loop is violated.
	Clearly, $S_I$ (resp. $S_J$) reconfigures $I$ (resp. $J$) to $I^\prime$.
	Therefore, $S = S_I \oplus \text{rev}(S_J)$ reconfigures $I$ to $J$.
	
	Next, we claim that $S$ is shortest.
	Note that any $\mathsf{TJ}$-sequence between $I$ and $J$ uses at least $|I\Delta J|/2$ $\mathsf{TJ}$-steps.
	Moreover, in $\mathtt{PVCRPathTJ}(P,I,J)$, we move tokens exactly $|I\Delta J|/2$ times: in each iteration, exactly one token (either from $I \setminus J$ or $J \setminus I$) is moved, and then the size of $I \Delta J$ decreases by $2$.
	Therefore, $S$ is shortest. 
	Consequently, the running time is $O(n)$.
\end{proof}

By Theorem~\ref{thm:kPVCR_paths_TJ} and Lemma~\ref{lem:TJ-TAR-equiv}, we obtain the following result on \textsc{$k$-PVCR} on a path $P$ under $\mathsf{TAR}$. 

\begin{theorem}
	\label{thm:kPVCR_paths_TAR}
	For any instance $(P,I,J, \mathsf{TAR}(u))$ of \textsc{$k$-PVCR} on a path $P$ on $n$ vertices, one can decide if $I$ and $J$ are reconfigurable in linear time. 
\end{theorem}
\begin{proof}
	A similar approach as in the proof of Theorem~\ref{thm:kPVCR-Trees-TAR} can be applied.
	Note that in the case $s = u$, where $s = |I| = |J|$, we have to check if we can remove at least one token from $I$ (resp. $J$) is as follows. 
	Given a path $P=v_{1}v_{2}\dots v_{n}$, let us assume that $I=\{v_{i_1}, v_{i_2}, \dots, v_{i_s}\}$ where $1\leq i_1 < i_2 < \dots < i_s\leq n$. In order to check if a token on $u$ can be removed, assuming $u=v_{i_j}$ for some $j$ such that $1\leq j\leq s$, we do as follows. (1) If $j\in\{2,\dots,s-1\}$, then check if $\text{dist}_G(v_{i_{j-1}},v_{i_{j+1}})\leq k$, and (2) if $j=1$, then check if $\text{dist}_G(v_{1},v_{i_{j}})\leq k-1$, and (3) if $j=s$, then check if $\text{dist}_G(v_{i_{j}},v_{n})\leq k-1$. 
	Indeed, this can be done in $O(n)$ time: for each token, one needs $O(1)$ time for checking if the resulting set obtained by removing $u$ is still a $k$-path vertex cover. The correctness of this checking easily follows from the definition of $k$-path vertex cover. One can see that similar things can be done for $J$. 
\end{proof}

Now we sketch the idea for solving the problem under $\mathsf{TS}$ in polynomial time. Given an instance $(P,I,J,\mathsf{TS})$ of \textsc{$k$-PVCR} where $|I|=|J|=s$, one can construct a shortest $\mathsf{TS}$-sequence between $I$ and $J$. Suppose that vertices in $I = \{v_{i_1}, \dots, \allowbreak v_{i_s}\}$ and $J = \{v_{j_1}, \dots, v_{j_s}\}$ are ordered such that $1 \leq i_1 < \dots < i_s \leq n$ and  $1 \leq j_1 < \dots < j_s \leq n$. Our goal is to construct a shortest $\mathsf{TS}$-sequence (of length $\sum_{p=1}^s\text{dist}_P(v_{i_p}, v_{j_p})$) that repeatedly slides the token on the ``leftmost'' vertex $v_{i_p} \in I$ to the ``leftmost'' vertex $v_{j_p} \in J$ if $i_p < j_p$ or vice-versa otherwise, for $p \in \{1, \dots, s\}$. 

The key point is, in certain conditions, one can construct in polynomial time a function $\mathtt{Push}(P, I, i, j)$ (Function~\ref{func:push}) whose task is to output a $\mathsf{TS}$-sequence that moves the token placed at some vertex $v_i$ of the $k$-path vertex cover $I$ to vertex $v_j$ in a given path $P = v_1v_2\dots v_n$, where $1 \leq i < j \leq i+k \leq n$.
Roughly speaking, $\mathtt{Push}(P, I, i, j)$ slides the token $t$ on $v_i$ toward $v_j$ along the path $P_{ij} = v_iv_{i+1}\dots v_j$ until either $t$ ends up at $v_j$ or there is some index $p \in \{i, \dots, j-1\}$ where $t$ is already placed at $v_{p}$ but cannot immediately move to $v_{p+1}$ because there is already some token $t^\prime$ placed there.
In the latter case, one can recursively call $\mathtt{Push}$ to slide $t^\prime$ from $v_{p+1}$ to $v_{p+2}$ and therefore enabling $t$ (which is currently placed at $v_{p}$) to slide to $v_{p+1}$. 
Now, the same situation happens again with $t$ and $t^\prime$, and the resolving procedure can be done in the same manner as before.
This process stops when $t$ is finally placed at $v_j$.

\begin{algorithm}[t]
	\SetAlgorithmName{Function}{}{}
	\KwIn{A path $P = v_1\dots v_n$, a $k$-path vertex cover $I$, and two indices $i$ and $j$ with $1 \leq i < j \leq i+k \leq n$\;}
	\KwOut{A sequence $S$ of $\mathsf{TS}$-steps that moves the token on $v_i$ to $v_j$\;}
	\SetArgSty{textbb}
	$S = \emptyset$\;
	\While{$i\neq j$}
	{
		\If{$v_{i+1}\in I$}
		{
			$S:=  S \oplus \mathtt{Push}(P, I, i+1, i+2)$\tcp*{Both $S$ and $I$ are updated}
		}
		$S:= S\oplus \langle I, I\setminus \{v_i\} \cup\{v_{i+1}\}\rangle$\;
		$I:= I\setminus \{v_i\} \cup \{v_{i+1}\}$\;
		$i:= i+1$\;
	}
	\Return $S$;
	\caption{$\mathtt{Push}(P, I, i, j)$}
	\label{func:push}
\end{algorithm}

The following lemma says that if certain conditions are satisfied, the output of $\mathtt{Push}(P, I, i, j)$ is indeed a $\mathsf{TS}$-sequence that reconfigures the $k$-path vertex cover $I$ to some other $k$-path vertex cover of $P$.
\begin{lemma}
	\label{lem:push}
	Let $P = v_1v_2\dots v_n$ be a path on $n$ vertices, and let $I$ be a $k$-path vertex cover of $P$.
	Let $i \in \{1, \dots, n\}$ be such that either $i \leq k+1$ or $\{v_{i-1}, \dots, v_{i-k}\} \cap I \neq \emptyset$.
	If $\{v_i, v_{i+1}, \dots, v_{i+p}\} \subseteq I$ and $v_{i+p+1} \notin I$ for some integer $p$ satisfying $0 \leq p \leq n-i-1$, then there exists a $\mathsf{TS}$-sequence in $P$ that reconfigures $I$ to $I \setminus \{v_i, v_{i+1}, \dots, v_{i+p}\} \cup \{v_{i+1}, \dots, v_{i+p+1}\}$.
	Consequently, if the assumption is satisfied, the output of $\mathtt{Push}(P, I, i, j)$ is indeed a $\mathsf{TS}$-sequence in $P$ that reconfigures $I$ to some $k$-path vertex cover of $P$.
\end{lemma}
\begin{proof}
	We prove the lemma by induction on $p$.
	If $p = 0$, then by the assumption, the lemma clearly holds because the token on $v_i$ can indeed be moved to $v_{i+1}$ without leaving any non-covered $k$-path.
	Assume that if $\{v_i, v_{i+1}, \dots, v_{i+p-1}\} \subseteq I$ and $v_{i+p} \notin I$ for some integer $p$ satisfying $0 \leq p \leq n-i-1$, then there exists a $\mathsf{TS}$-sequence $S^\prime$ in $P$ that reconfigures $I$ to $I \setminus \{v_i, v_{i+1}, \dots, v_{i+p-1}\} \cup \{v_{i+1}, \dots, v_{i+p}\}$.
	We claim that if $\{v_i, v_{i+1}, \dots, v_{i+p}\} \subseteq I$ and $v_{i+p+1} \notin I$ for some integer $p$ satisfying $0 \leq p \leq n-i-1$, then there exists a $\mathsf{TS}$-sequence $S$ in $P$ that reconfigures $I$ to $I \setminus \{v_i, v_{i+1}, \dots, v_{i+p}\} \cup \{v_{i+1}, \dots, v_{i+p+1}\}$.
	Note that the $k$-path $v_{i+p-k+1}\dots v_{i+p}$ is (at least) covered by both $v_{i+p-1}$ and $v_{i+p}$.
	Therefore, the token on $v_{i+p}$ can be slid to $v_{i+p+1}$ without leaving any non-covered $k$-path.
	More formally, $I^\prime = I \setminus \{v_{i+p}\} \cup \{v_{i+p+1}\}$ is a $k$-path vertex cover in $P$.
	By the inductive hypothesis, there exists a $\mathsf{TS}$-sequence $S^\prime$ that reconfigures $I^\prime$ to $I^\prime \setminus \{v_i, v_{i+1}, \dots, v_{i+p-1}\} \cup \{v_{i+1}, \dots, v_{i+p}\} = I \setminus \{v_i, v_{i+1}, \dots, v_{i+p}\} \cup \{v_{i+1}, \dots, v_{i+p+1}\}$.
	Thus, $S = \langle I, I^\prime \rangle \oplus S^\prime$ is our desired $\mathsf{TS}$-sequence.
	It is not hard to see that each iteration of the {\bf while} loop in $\mathtt{Push}(P, I, i, j)$ performs exactly the procedure we have just described (the case $p = 0$ corresponds to the steps outside the {\bf if} condition, the case $p \geq 0$ corresponds to the recursive call inside the {\bf if} condition).
	As a result, if the assumption of this lemma is satisfied, $\mathtt{Push}(P, I, i, j)$ is indeed a $\mathsf{TS}$-sequence.
\end{proof}

Clearly, the function $\mathtt{Push}(P, I, i_p, j_p)$ can be used to slide a token on $v_{i_p}$ to $v_{j_p}$ for $p\in \{1, \dots, s\}$ and $i_p < j_p$. Thus, we have the following theorem. 

\begin{theorem}
	\label{thm:kPVCR_paths_TS} 
	Given an instance $(P, I, J, \mathsf{TS})$ of \textsc{$k$-PVCR} on a path $P$, the $k$-path vertex covers $I$ and $J$ are reconfigurable if and only if $|I| = |J|$.
	Moreover, we can compute a shortest reconfiguration sequence in $O(n^2)$ time. 
\end{theorem}
\begin{proof}
	Before proving Theorem~\ref{thm:kPVCR_paths_TS}, we describe the algorithm $\mathtt{PVCRPathTS}(P,I,J)$ (Algorithm~\ref{algo:PVCRPathTS}) that takes two $k$-path vertex covers $I$ and $J$ of $P$ with $|I| = |J|$ as the input, and returns a $\mathsf{TS}$-sequence between them.
	In the following, we use the expression \emph{leftmost} instead of using ``with the smallest index''.
	
	Suppose that vertices in $I = \{v_{i_1}, \dots, \allowbreak v_{i_s}\}$ and $J = \{v_{j_1}, \dots, v_{j_s}\}$ are ordered such that $1 \leq i_1 < \dots < i_s \leq n$ and  $1 \leq j_1 < \dots < j_s \leq n$, where $s = |I| = |J|$.
	Intuitively, $\mathtt{PVCRPathTS}(P,I,J)$ outputs a $\mathsf{TS}$-sequence that slides the token on $v_{i_p}$ to $v_{j_p}$ for $p \in \{1, \dots, s\}$. 
	Since $P$ is a path, this is the only way of sliding tokens, and thus any $\mathsf{TS}$-sequence between $I$ and $J$ uses at least $\sum_{p=1}^s\text{dist}_P(v_{i_p}, v_{j_p})$ $\mathsf{TS}$-steps.
	
	\begin{algorithm}[b]
		\KwIn{A path $P = v_1v_2\dots v_n$, two $k$-path vertex covers $I, J$\;}
		\KwOut{A $\mathsf{TS}$-sequence $S$ between $I$ and $J$ in $P$\;}
		\SetArgSty{textbb}
		Let $S, S_I, S_J$ be reconfiguration sequences, and initialize them by $\emptyset$\;
		\While{$I \neq J$}{
			Mark all vertices in $I$ and $J$ as $\mathtt{untouched}$\;
			Find the leftmost $\mathtt{untouched}$ vertex $v_i \in I$ and the leftmost $\mathtt{untouched}$ vertex $v_j \in J$\;
			\If{$i < j$}{
				$S_I := S_I \oplus \mathtt{Push}(P, I, i, j)$ \tcp*{$I$ is updated in Push}
			}
			\Else{
				$S_J := S_J \oplus \mathtt{Push}(P, J, j, i)$ \tcp*{$J$ is updated in Push}
			}
			Mark $v_i$ and $v_j$ as $\mathtt{touched}$\;
		}
		$S: = S_I \oplus \text{rev}(S_J)$\;
		\Return $S$\;
		\caption{$\mathtt{PVCRPathTS}(P,I,J)$}
		\label{algo:PVCRPathTS}
	\end{algorithm}
	
	Now we prove Theorem~\ref{thm:kPVCR_paths_TS}. As before, the only-if direction is trivial. 
	We show that $\mathtt{PVCRPathTS}(P,I,J)$ constructs a shortest $\mathsf{TS}$-sequence between two $k$-path vertex covers $I, J$ of $P$ with $|I| = |J|$ in $O(n^2)$ time.
	
	We first verify that the output of $\mathtt{PVCRPathTS}(P,I,J)$ is a $\mathsf{TS}$-sequence between $I$ and $J$ in $P$.
	Note that if in the current iteration of the {\bf while} loop in $\mathtt{PVCRPathTS}$, the token on $v_i$ is moved to $v_j$ (i.e., $i < j$), then the distance between $v_j$ and the two $\mathtt{untouched}$ vertices considered in the next iteration must be at most $k$; otherwise, some non-covered $k$-path appears.
	Then, the assumption of Lemma~\ref{lem:push} is satisfied in the next iteration.
	A similar argument holds for $i > j$.
	As a result, the function $\mathtt{Push}$ always returns a $\mathsf{TS}$-sequence.
	Let $I^\prime$ be the $k$-path vertex cover of $P$ obtained when the condition of the {\bf while} loop of $\mathtt{PVCRPathTS}(P,I,J)$ is violated.
	Then, it is not hard to see that $S_I$ (resp. $S_J$) is a $\mathsf{TS}$-sequence that reconfigures $I$ (resp. $J$) to $I^\prime$, and therefore $S = S_I \oplus \text{rev}(S_J)$ reconfigures $I$ to $J$. 
	
	Note that in the function $\mathtt{Push}(P, I, i, j)$ (and also $\mathtt{Push}(P, J, j, i)$), $\mathtt{Push}$ is called at most once for each vertex of $P$, which implies $\mathtt{Push}(P, I, i, j)$ runs in $O(n)$ time. 
	Moreover, $\mathtt{PVCRPathTS}$ marks each vertex in $I$ and $J$ exactly twice.
	Thus, in total, $\mathtt{PVCRPathTS}$ runs in $O(n^2)$ time.
	
	To conclude the proof of Theorem~\ref{thm:kPVCR_paths_TS}, we show that the $\mathsf{TS}$-sequence $S$ between $I$ and $J$ in $P$ obtained from $\mathtt{PVCRPathTS}(P,I,J)$ is shortest.
	To see this, note that for each $p \in \{1, \dots, s\}$, either the token $t$ on $v_{i_p} \in I$ is slid to $v_{j_p} \in J$ or the token $t^\prime$ on $v_{j_p} \in J$ is slid to $v_{i_p} \in I$ in some iteration of the {\bf while} loop in $\mathtt{PVCRPathTS}(P,I,J)$, and either $S_I$ or $S_J$ is then updated accordingly.
	Suppose that the algorithm slides $t$ to $v_{j_p}$.
	Note that if there is any token $t^{\prime\prime}$ placed at some vertex $v_{i_q}$ ($i_q \in \{i_p+1, \dots, j_p\}$ in the path $v_{i_p}v_{i_p+1}\dots v_{j_p}$, then even when $t^{\prime\prime}$ is moved by some $\mathtt{Push}$ calls, by the time $t$ ends up at $v_j$, $t^{\prime\prime}$ cannot be placed at any vertex whose index is larger than $j_q$.
	(We always have $i_p < i_q \leq j_p < j_q$ for all such $i_q$.)
	Clearly, if no such $v_{i_q}$ exists, sliding $t$ has no effect on sliding any other token in the next iterations.
	A similar argument holds in case the algorithm slides $t^\prime$.
	Thus, we can conclude that $\mathtt{PVCRPathTS}(P,I,J)$ performs exactly $\sum_{p=1}^s\text{dist}_P(v_{i_p}, v_{j_p})$ $\mathsf{TS}$-steps, and therefore outputs a shortest $\mathsf{TS}$-sequence.
\end{proof}

\subsubsection{\textsc{$k$-PVCR} on Cycles.}
Let $C=v_0v_1\dots v_{n-1}v_0$ be a given $n$-vertex cycle, and let $(C,I,J, \textsf{R})$ be a \textsc{$k$-PVCR} instance on $C$ under a reconfiguration rule $\textsf{R}\in \{\mathsf{TJ}, \mathsf{TS}, \mathsf{TAR}(u)\}$. We remark that if $|I|=|J|=\lceil n/k\rceil$ and $n=c\cdot k$ for some $c$, then $(C,I,J, \mathsf{R})$ where $\mathsf{R}\in \{\mathsf{TS}, \mathsf{TJ}\}$ is a no-instance. This is because no tokens can be moved in such instances. 

Here we assume that the indices of vertices on the cycle increase in the clockwise manner. 
We claim that it is possible to apply the algorithms for paths to cycles, by cutting a cycle into a path with a vertex in $I\cap J$, \red{if it exists}.
\red{Moreover, if $I \cap J = \emptyset$, we claim that one can always move tokens to reach an instance where $I \cap J \neq \emptyset$.}
Our algorithms do not always achieve the shortest reconfiguration sequence. However, we later show that achieving the shortest sequence even on cycles under $\mathsf{TJ}$ might not be trivially easy, since we can systematically create the instances such that the length of the shortest reconfiguration sequence is not equal to $|I\Delta J|/2$.  

Now, we describe the sketch how to cut $C$ under $\mathsf{TJ}$, $\mathsf{TS}$, and $\mathsf{TAR}$. In the $\mathsf{TS}$ case, without loss of generality, we can assume that either $|I|\neq \lceil n/k\rceil$ or $n\neq c\cdot k$ holds. If $v$ is already in $I\cap J$, we cut $C$ by removing $v$. The following lemma ensures that if $I$ and $J$ are reconfigurable in $C-v$, then $I\cup\{v\}$ and $J\cup\{v\}$ are reconfigurable in $C$.

\begin{lemma}
	\label{lemma:cut_cycle_safe}
	Let $C$ be an $n$-vertex cycle and $v$ be a token in $I\cap J$ of $C$. Then, for any $k$-path vertex cover $I'$ of $C-v$, $I'\cup \{v\}$ is a $k$-path vertex cover of $C$. 
\end{lemma}
\begin{proof}
	Let us assume $v$ to be $v_0$.
	Consider \red{the} path $P=C-v= v_1v_2\dots v_{n-2}v_{n-1}$ and a $k$-path vertex cover $I'$ on $P$. Since $I'$ covers all the $k$-paths on $P$, $I'$ has at least one token on the $k$-path $P'=v_1v_2\dots v_k$ and also at least one token on the $k$-path $P''=v_{n-k}v_{n-k+1}\dots v_{n-1}$. Now $v$ is a token in $I\cap J$, if we connect two endpoints $v_1$ and $v_{n-1}$ with $v$ and create a cycle, all new $k$-paths include $v$ and those paths are covered by $v$. This completes the proof.
\end{proof}

If $I\cap J =\emptyset$, there exists at least one token movable in the clockwise or counterclockwise direction. Here, we say a token $u$ is \emph{movable} if and only if (i) there exists a neighbor $v$ of $u$ such that no token is placed on $v$, and (ii) moving a token on $u$ to $v$ results a $k$-path vertex cover. 

\begin{lemma}
	\label{lem:token_movable}
	If either $|I|\neq \lceil n/k\rceil$ or $n\neq c\cdot k$ holds, then there exists at least one token movable by at least one step in the clockwise or counterclockwise direction. Furthermore, we can find such a token in linear time. 
\end{lemma}
\begin{proof}
	If $|I|\neq \lceil n/k \rceil$, since $\lceil n/k \rceil$ is a minimum size of $k$-path vertex cover on $n$-vertex cycle, we can assume $|I|\geq \lceil n/k \rceil+1$. This implies that there exists some $k$-path that has at least two tokens on it. We can find such a path (and thus such tokens) in linear time, since there are at most $n$ distinct $k$-paths on an $n$-vertex cycle. Once we find such tokens, e.g., $u$ and $v$, at least one of them can move at least one step in clockwise or counterclockwise direction, since the $k$-path is now covered by $u$ and $v$ and if we move $v$, either $u$ or $v$ still covers the $k$-path. Hence, if $|I|\neq \lceil n/k \rceil$, this lemma holds. 
	
	Consider the case $|I|=\lceil n/k \rceil$ and $n$ is not divisible by $k$. Since $I$ is a $k$-path vertex cover, $I$ covers all $k$-paths in $C$. Clearly, $C$ is a cycle of size $n$ if and only if the number of edges of $C$ is $n$. Suppose to the contrary that each $k$-path in $C$ has exactly one token of $I$. Then, the length of the cycle is $|I|\cdot (k-1)+|I|=|I|\cdot k$, which contradicts to the assumption that $n$ is not divisible by $k$. By this argument, similarly to the above $|I|\neq \lceil n/k \rceil$ case, there exists at least one $k$-path which has two tokens of $I$, and we can find them in linear time. This completes the proof.
\end{proof}

After finding such a movable token, we can use \emph{rotate} operation repeatedly until \red{we obtain} at least one vertex in $I\cap J$. Here, the rotate operation takes a token-set, a movable token \red{$u$} which can be slid at least one step towards direction $d\in \{\text{clockwise}, \text{counterclockwise}\}$ as input, and outputs a $\mathsf{TS}$-sequence that slides all tokens one step towards $d$. 
\red{Intuitively, moving $u$ one step towards $d$ enables its successor (with respect to direction $d$) to move one step towards $d$, and so on.}
After obtaining at least one vertex in $I\cap J$, we can perform the cutting operation as before. 

Next, we consider the $\mathsf{TJ}$ case. Since any $\mathsf{TS}$-sequence is also a $\mathsf{TJ}$-sequence, we can perform the same cutting operation as in the $\mathsf{TS}$ case. Then, using this cutting operation, we can show the following theorem.

\begin{theorem}
	\label{thm:kPVCR_cycle_TJTS} 
	Given an instance $(C, I, J, \textsf{R})$ of \textsc{$k$-PVCR} on a cycle $C$ where $\textsf{R}\in \{
	\textsf{TS}, \textsf{TJ}\}$, if $|I|=|J|=\lceil n/k \rceil$ and $n=c
	\cdot k$ for some $c$, then $(C,I,J, \textsf{R})$ is a no-instance. Otherwise, the $k$-path vertex covers $I$ and $J$ are reconfigurable if and only if $|I|=|J|$. Moreover, we can compute a reconfiguration sequence for $\mathsf{TJ}$ rule in $O(n)$ time, and for $\mathsf{TS}$ rule in $O(n^2)$ time.
\end{theorem}
\begin{proof}
	We describe an algorithm (Algorithm~\ref{algo:PVCRCycleTS}) that takes $C=v_0v_1\dots v_{n-1}v_0$, initial token-set $I$, and target token-set $J$ and outputs a reconfiguration sequence $S$ if exists, and otherwise says no-instance. Lemma~\ref{lemma:cut_cycle_safe} shows that it is possible to cut the cycle $C$ with a vertex $v\in I\cap J$; in other words, it is equivalent to consider problems on a path $P=C-v$. 
	
	Lemma~\ref{lem:token_movable} allows us to find at least one movable token if either $|I|\neq \lceil n/k\rceil$ or $n\neq c\cdot k$ holds. 
%
	After finding such a movable token, we can use the \emph{rotate} operation described in Function~\ref{func:rot} 
	and obtain at least one vertex $v\in I\cap J$. Let us assume that $I=\{v_{i_{0}}, v_{i_{1}}, \dots, v_{i_{s-1}}\}$ where $0 \leq i_{0} < i_{1} < \dots < i_{s-1} \leq n-1$. Here, let $v_{i_{j}}$ be a token that is movable to at least one step in the direction $d\in \{\text{clockwise}, \text{counterclockwise}\}$, where $j\in \{0,\dots,s-1\}$. 
	
	\begin{algorithm}[b]
		\SetAlgorithmName{Function}{}{}
		\KwIn{A token-set $I$, a token on $v_{i_{j}}\in I$, $d\in \{\text{clockwise}, \text{counterclockwise}\}$.}
		\KwOut{A $\mathsf{TS}$-sequence $S$ \red{that slides all tokens one step towards $d$, starting from $v_{i_j}$}.}
		\SetArgSty{textbb}
		$S:=\emptyset$\;
		$c:=j$\;
		\If{$d$ is clockwise}
		{
			\Repeat{$c=j$}
			{
				$S:=S\oplus \langle I,I\setminus\{v_{i_{c}}\}\cup \{v_{(i_{c}+1) \mod n}\}\rangle$\;
				$I:=I\setminus \{v_{i_{c}}\}\cup \{v_{(i_{c}+1) \mod n}\}$\;
				$c:=(c+1)\mod \red{|I|}$\;
			}
		}
		\Else
		{
			\Repeat{$c=j$}
			{
				$S:=S\oplus \langle I,I\setminus\{v_{i_{c}}\}\cup \{v_{(i_{c}-1) \mod n}\}\rangle$\;
				$I:=I\setminus \{v_{i_{c}}\}\cup \{v_{(i_{c}-1) \mod n}\}$\;
				$c:=(c-1)\mod \red{|I|}$\;
			}
		}
		\Return $S$;
		\caption{$\mathtt{rot}(I, i_j, d)$}
		\label{func:rot}
	\end{algorithm}
	
	One can observe that, by Lemma~\ref{lem:token_movable}, the reconfiguration sequence obtained by $\mathtt{rot}(I, i_j, d)$ is a $\mathsf{TS}$-sequence. This is indeed true, since it moves each token in $I$ by exactly one step keeping the $k$-path vertex cover property, by starting to move tokens from $v_{i_j}$ along the cycle until we meet $v_{i_j}$ again. 
	
	\begin{algorithm}[t]
		\KwIn{A cycle $C=v_0v_1\dots v_{n-1}v_0$, initial token-set $I$, and target token-set $J$\;}
		\KwOut{A reconfiguration sequence $S$ if it exists; otherwise says no-instance\;}
		\SetArgSty{textbb}
		$S:=\emptyset$\;
		\If{$I\cap J=\emptyset$ and $|I|=\lceil n/k \rceil$ and $n$ is divisible by $k$}
		{
			\Return $(C,I,J)$ is a no-instance\;
		}
		Find a token $v_{i}\in I$ such that it can move at least one step in clockwise or counterclockwise direction, and let $d$ be such  a direction\;
		\While{$I\cap J=\emptyset$}
		{
			$S:=S\oplus \mathtt{rot}(I, i, d)$ \tcp*{$I$ is updated in rot($I,i,d$)}
			\If{d is clockwise}
			{
				$i:=(i+1)\mod n$\;
			}
			\Else
			{
				$i:=(i-1)\mod n$\;
			}
		}
		Pick one token $v\in I\cap J$\; 
		$S^\prime = \mathtt{PVCRPathTS}(C-v,I,J)$\;
		Update $S^\prime$ by adding $v$ to each of its members\;
		$S:=S~\oplus S^\prime$\;
		\Return $S$\;
		\caption{$\mathtt{PVCRCycleTS}(C,I,J)$}
		\label{algo:PVCRCycleTS}
	\end{algorithm}
	
	By Lemma~\ref{lemma:cut_cycle_safe} and Lemma~\ref{lem:token_movable}, \texttt{PVCRCycleTS$(C,I,J)$} is shown to be correct. Note here that, for \textsc{$k$-PVCR} on cycles under $\mathsf{TJ}$, one can use \texttt{PVCRPathTJ$(C-v,I,J)$} instead of applying \texttt{PVCRPathTS$(C-v,I,J)$} in the algorithm. 
	For the computation time, since (i) \textbf{while} loop takes $O(kn)$ time and (ii) \texttt{PVCRPathTS$(C-v,I,J)$} runs in $O(n^2)$ time, \texttt{PVCRCycleTS$(C,I,J)$} runs in $O(n^2)$ time. For $\mathsf{TJ}$ case, since \texttt{PVCRPathTJ$(C-v,I,J)$} runs in $O(n)$ time, \texttt{PVCRCycleTJ$(C,I,J)$} runs in $O(n)$ time.
\end{proof}

For the $\mathsf{TAR}$ case, we can use the result for the $\mathsf{TJ}$ case and Lemma~\ref{lem:TJ-TAR-equiv} to show the following theorem.

\begin{theorem}
	\label{thm:kPVCR_cycles_TAR}
	For any instance $(C,I,J, \mathsf{TAR}(u))$ of \textsc{$k$-PVCR} on a cycle $C$, one can decide if $I$ and $J$ are reconfigurable in linear time. 
\end{theorem}
\begin{proof}
	Clearly, if $u<\max \{|I|,|J|\}$ or $u=\psi_k(C)$ then $(C,I,J, \mathsf{TAR})$ is a no-instance, because either $I$ or $J$ cannot be modified by adding/removing tokens. We now consider the case $u\geq \max\{|I|,|J|\}$ and $u>\psi_k(C)$. Note that if $|I|<|J|$ then we can add tokens to $I$ until the resulting $k$-path vertex cover is of size $|J|$, simply because $u\geq \max\{|I|,|J|\}$. As a result, we can assume without loss of generality that $|I|=|J|=s$ for some constant $s$. Now we have $|I|=|J|=s$ and $u>\psi_k(C)$, we divide into two cases: $u\geq s+2$ or $u=s+1$.
	
	Consider the case $u\geq s+2$. If $I\cap J=\emptyset$, then we add one token $v\notin I\cup J$. Then we can cut $C$ by $v$ and consider the instance $(C - v, I \setminus \{v\}, J \setminus \{v\})$ on the path $C - v$ under $\mathsf{TAR}(u^\prime)$ where $u^\prime \geq s+1$. Then, by Theorem~\ref{thm:kPVCR_paths_TJ} and Lemma~\ref{lem:TJ-TAR-equiv}, $I$ is always reconfigurable to $J$ under \textsf{TAR}$(u)$. Otherwise, i.e., $I\cap J\neq \emptyset$, we can use the similar argument as before with $|I|=|J|=s-1$ and $u\geq s+1$, therefore $I$ is always reconfigurable to $J$ under \textsf{TAR}$(u)$. 
	
	Next, consider the case $u=s+1$. If $I\cap J\neq \emptyset$, also similar argument can be applied as before with $|I|=|J|=s-1$ and $u=s$. Hence, in this case, $I$ is always reconfigurable to $J$ under \textsf{TAR}$(u)$. Otherwise, we first find a token of $I$ which is movable in direction $d \in \{\text{clockwise}, \text{counterclockwise}\}$. Recall that a token $u$ is movable to some vertex $v$ if the resulting set still keeps a $k$-path vertex cover property. If we can find such a token, we can rotate $I$ in $d$ until $I\cap J\neq \emptyset$ as in the algorithm $\mathtt{PVCRCycleTS}$. We note that though such rotation forms a $\mathsf{TS}$-sequence (which is also a $\mathsf{TJ}$-sequence), by Lemma~\ref{lem:TJ-TAR-equiv}, it can be converted to a \textsf{TAR}$(u)$-sequence. If we finish the rotation, then we can also cut $C$ by $v\in I\cap J$ and similar argument can be applied as before. Else, assume without loss of generality that no token in $I$ can move. Then, by Lemma~\ref{lem:token_movable}, it follows that $I$ is minimum and $n=c\cdot k$. Now we have $u=s+1$, and we can add exactly one token. However, even when adding a new token $v$, one cannot remove any other token $u$ while keeping the $k$-path vertex cover property. Suppose to the contrary, let $I^\prime = I\setminus \{u\}\cup \{v\}$. This implies that $I^\prime$ can be obtained from $I$ by jumping the token on $u$ to $v$. However, since $n=c\cdot k$ and $I$ and $I^\prime$ are token sets of minimum size, then $I$ cannot be reconfigured to $I^\prime$ under $\mathsf{TJ}$, a contradiction. 
\end{proof}

So far, we have shown that \textsc{$k$-PVCR} on cycles under each of $\mathsf{TJ}$ and $\mathsf{TAR}(u)$ can be solved in $O(n)$ time, and under $\mathsf{TS}$ can be solved in $O(n^2)$ time (Theorems~\ref{thm:kPVCR_cycle_TJTS} and~\ref{thm:kPVCR_cycles_TAR}).
To conclude this section, we give an example showing that even in a yes-instance $(C, I, J, \mathsf{TJ})$ of \textsc{$k$-PVCR} ($k \geq 3$) under $\mathsf{TJ}$ on a cycle $C$, one may need to use more than $|I \Delta J|/2$ $\mathsf{TJ}$-steps even in a shortest $\mathsf{TJ}$-sequence.
Intuitively, the lower bound $|I \Delta J|/2$ seems to be easy to achieve under $\mathsf{TJ}$, simply by jumping tokens one by one from $I \setminus J$ to $J \setminus I$.
However, as we show in the following lemma, to keep the $k$-path vertex cover property, sometimes a token in $I$ may need to jump to some vertex not in $J \setminus I$ beforehand.
This implies the non-triviality of finding a \emph{shortest} reconfiguration sequence even under $\mathsf{TJ}$. 

\begin{lemma}
	\label{lemma:cycle_notshortestTJ}
	For \textsc{$k$-PVCR} ($k\geq3$) yes-instances $(C, I, J, \mathsf{TJ})$ on cycles where $C=v_0v_1\dots v_{3k-2}v_0$, $I=\{v_0, v_k, v_{2k}\}$ and $J=\{v_{3k-2}, v_{2k-2}, \red{v_{k-1}}\}$, the length of a shortest reconfiguration sequence from $I$ to $J$ is greater than $|I\Delta J|/2$. 
\end{lemma}
\begin{figure}[b]
	\centering
	\includegraphics[scale=0.25]{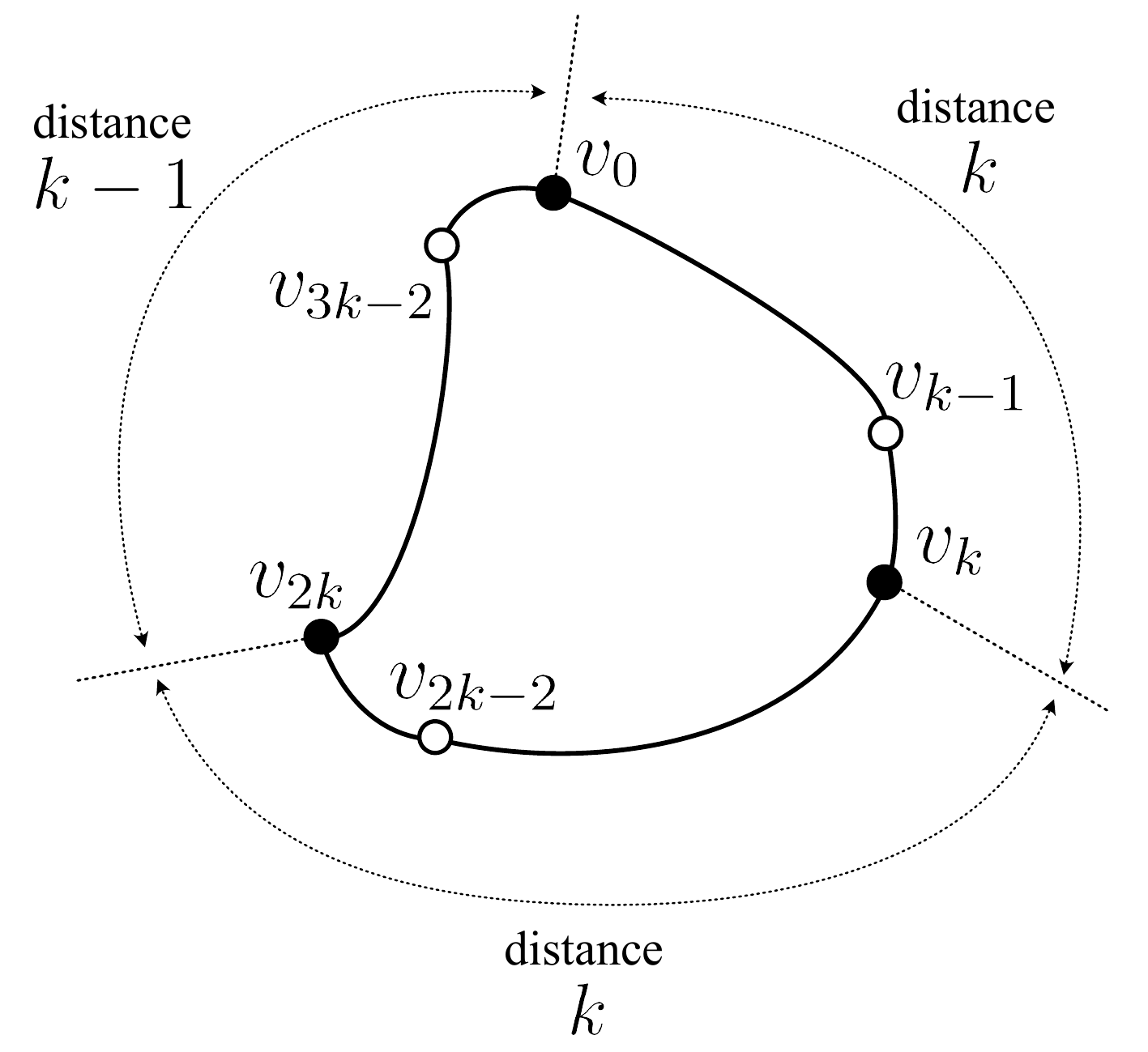}
	\caption{An instance $(C,I,J, \mathsf{TJ})$ that requires more than $|I\Delta J|/2$ steps to reconfigure}
	\label{fig:cycle_of_lemma:cycle_notshortestTJ}
\end{figure}

\begin{proof}
	We illustrate such instances in \figurename~\ref{fig:cycle_of_lemma:cycle_notshortestTJ}. In \figurename~\ref{fig:cycle_of_lemma:cycle_notshortestTJ}, black tokens are in $I$, and white tokens are in $J$. Note that $\text{dist}_C(v_{2k-2}, v_{2k}) = 2$ and $\text{dist}_C(v_{3k-2}, v_{0}) = \red{\text{dist}_C(v_{k-1}, v_{k})} = 1$. 
	
	First, $v_0$ is the only vertex that covers the path $P = v_0v_1\dots v_{k-1}$, which means $v_0$ cannot move to some vertex outside $P$, such as $v_{3k-2}$. Therefore, $v_0$ has no choice but to move to \red{$v_{k-1}$}. However then, the path $v_{2k}\dots v_{3k-2}v_0\dots \red{v_{k-1}}$ is of length $\red{2k-2}\geq k$. By these arguments, $v_0$ cannot directly move to \red{$v_{k-1}$}. Similarly, since $v_{2k}$ is the only vertex that covers the path $P' = \red{v_{k+1}}\dots v_{2k}$,  the possible way is only to move $v_{2k}$ to $v_{2k-2}$, which also results in an non-covered path $v_{2k-1}\dots v_{3k-2}$ of length $k-1$. It is clear that $v_{k}$ cannot move \red{directly to} either $v_{2k-2}$ or \red{$v_{k-1}$}. Therefore, every token in $I$ cannot move directly to one of the tokens in $J$, which means it requires at least one step to put some token on some vertex $v\notin I\Delta J$. This also holds for the case moving tokens in $J$ to $I$. Hence, the length of the reconfiguration sequence is greater than $|I\setminus J|=|J\setminus I|=|I\Delta J|/2$. 
	
	Finally, we confirm that the created instance is a yes-instance. First, for example, one can move $v_{2k}$ to $v_{2k-1}$, since after such a move the $k$-vertex path $v_{2k}\dots v_0$ is covered by the token $v_0$ and another $k$-vertex path $v_{2k-1}\dots v_{3k-2}$ is covered by the token $v_{2k-1}$. Then, now the length of path $v_{2k-1}\dots v_{k}$ is $k$, hence \red{$v_k$} can be moved to \red{$v_{k-1}$} by the similar argument. Therefore, by the reconfiguration sequence $S = \langle I = \{v_0, v_k, v_{2k}\}, \{v_0, v_k, v_{2k-1}\}, \{v_0, v_{k-1}, v_{2k-1}\}, \allowbreak\{v_{3k-2}, v_{k-1}, v_{2k-1}\}, \{v_{3k-2}, \red{v_{k-1}}, v_{2k-2}\}=J \rangle$, one can reconfigure $I$ to $J$. 
\end{proof}

\section{Concluding Remarks}
\label{sec:conclusion}
In this paper, we have investigated the complexity of \textsc{$k$-PVCR} under each of $\mathsf{TS}$, $\mathsf{TJ}$, and $\mathsf{TAR}$ for several graph classes.
In particular, several known hardness results for \textsc{VCR} (i.e., $k = 2$) can be generalized for \textsc{$k$-PVCR} when $k \geq 3$.
Additionally, we proved a complexity dichotomy for \textsc{$k$-PVCR} by showing that it remains $\mathtt{PSPACE}$-complete even if the input (planar) graph is of maximum degree \red{three} (using a reduction from NCL) and can be solved in polynomial time when the input (planar) graph is of maximum degree \red{two} (i.e., it is either a path or a cycle). 
On the positive side, we designed polynomial-time algorithms for \textsc{$k$-PVCR} on trees under each of $\mathsf{TJ}$ and $\mathsf{TAR}$.
We also showed how to construct a \emph{shortest} reconfiguration sequence on paths, and presented an example showing the nontriviality of finding shortest reconfiguration sequences on cycles even under $\mathsf{TJ}$.
The question of whether one can solve \textsc{$k$-PVCR} on trees under $\mathsf{TS}$ in polynomial time remains open.
Another target graphs may be chordal graphs (under each of $\mathsf{TJ}$ and $\mathsf{TAR}$), cographs, and graphs of treewidth at most $2$.
Even on graphs of treewidth at most $2$ (and moreover, on outerplanar graphs), the complexity of \textsc{VCR} remains open.

\section*{Acknowledgements}

\red{We would like to thank the anonymous reviewers for their insightful and valuable comments that help improving the preliminary versions of this paper.}
This work is partially supported by JSPS KAKENHI Grant Numbers JP20H05964 (D.A. Hoang), JP18H04091, JP20K11666 and JP20H05794 (A. Suzuki), Japan.

\bibliographystyle{plain}
\bibliography{refs.bib}

\end{document}